%% file: nonholonomic_contact.tex
\documentclass[12pt]{article}

\usepackage{wrapfig}
\usepackage{graphicx}
\usepackage[none]{hyphenat}
\usepackage[english]{babel}
\usepackage[utf8]{inputenc}
\usepackage[T1]{fontenc}

\usepackage{authblk}
\usepackage{tikz-cd}

\usepackage{caption}
\usepackage{csquotes}
\usepackage{mathtools,amssymb,amsfonts,amsthm}

\usepackage[backend=biber,giveninits=true,maxnames=5,doi=true,isbn=false,url=false,sortcites,style=phys,biblabel=brackets,sorting=nyt]{biblatex}
\usepackage[breaklinks,unicode=true]{hyperref}
\usepackage[capitalise]{cleveref} %

\theoremstyle{plain}
\newtheorem{theorem}{Theorem}
\newtheorem*{theorem*}{Theorem}
\newtheorem{lemma}[theorem]{Lemma}
\newtheorem*{lemma*}{Lemma}
\newtheorem{proposition}[theorem]{Proposition}
\newtheorem*{proposition*}{Proposition}
\newtheorem{corollary}[theorem]{Corollary}
\newtheorem*{corollary*}{Corollary}

\theoremstyle{definition}
\newtheorem{definition}{Definition}
\newtheorem*{definition*}{Definition}
\newtheorem{example}{Example}
\newtheorem*{example*}{Example}

\crefname{theorem}{Theorem}{Theorems}
\Crefname{theorem}{Theorem}{Theorems}
\crefname{lemma}{Lemma}{Lemmas}
\Crefname{lemma}{Lemma}{Lemmas}
\crefname{proposition}{Proposition}{Propositions}
\Crefname{Prop}{Proposition}{Propositions}
\crefname{corollary}{Corollary}{Corollaries}
\Crefname{corollary}{Corollary}{Corollaries}
\crefname{definition}{Definition}{Definitions}
\Crefname{definition}{Definition}{Definitions}
\crefname{example}{Example}{Examples}
\Crefname{example}{Example}{Examples}
\crefname{theorem}{Theorem}{Theorems}

\newtheorem*{exercise*}{Exercise}
\crefname{exercise}{exercise}{exercises}
\Crefname{exercise}{Exercise}{Exercises}  

\theoremstyle{remark}
\newtheorem{remarkx}{Remark}
\newtheorem*{remarkx*}{Remark}
\crefname{remark}{Remark}{Remarks}
\Crefname{remark}{Remark}{Remarks}
\newenvironment{remark}
  {\pushQED{\qed}\remarkx}
  {\popQED\endremarkx}
\newenvironment{remark*}
  {\pushQED{\qed}\remarkx*}
  {\popQED\endremarkx*}

\newcommand{\dv}[2]{\frac{\dd #1}{\dd #2}}
\newcommand{\pdv}[2]{\frac{\partial #1}{\partial #2}}

\DeclarePairedDelimiter{\set}{\lbrace}{\rbrace}

\newcommand{\Cont}{\mathcal{C}}
\newcommand{\RR}{\mathbb{R}}
\newcommand{\dd}{\mathrm{d}}

\newcommand{\Reeb}{\mathcal{R}}
\newcommand*{\ann}[1]{{#1}^{\circ}} %

\newcommand*{\contr}[1]{\iota_{#1}}
\newcommand*{\VecFields}{\mathfrak{X}}

\newcommand*{\lieD}[1]{\mathcal{L}_{#1}}

\newcommand*{\orth}[1]{{#1}^{\bot}}
\newcommand*{\orthL}[1]{\prescript{\bot}{}{#1}}

\DeclareMathOperator{\supp}{supp}

\DeclareMathOperator{\pr}{pr}

\DeclarePairedDelimiter{\gen}{\langle}{\rangle}
\DeclarePairedDelimiter{\lieBr}{\lbrack}{\rbrack}
\DeclarePairedDelimiter{\jacBr}{\lbrace}{\rbrace}
\DeclarePairedDelimiter{\NHBr}{\lbrace}{\rbrace_{L,\Delta}}

\title{Contact Hamiltonian systems with nonholonomic constraints}

\author[+*]{Manuel de León}
\author[*]{V\'ictor Manuel Jim\'enez}
\author[*]{Manuel Lainz Valcázar}

\affil[*]{Instituto de Ciencias Matem\'aticas (CSIC-UAM-UC3M-UCM),
C\textbackslash Nicol\'as Cabrera, 13-15, Campus Cantoblanco, UAM
28049 Madrid, Spain} 

\affil[+]{Real Academia de Ciencias Exactas, Fisicas y Naturales, C/de Valverde
22, 28004 Madrid, Spain}
\date{\today}

\begin{document}

\maketitle

\begin{abstract}
    In this article we develop a theory of contact systems with nonholonomic constraints. We obtain the dynamics from Herglotz's variational principle, by restricting the variations so that they satisfy the nonholonomic constraints. We prove that the nonholonomic dynamics can be obtained as a projection of the unconstrained Hamiltonian vector field. Finally, we construct the nonholonomic bracket, which is an almost Jacobi bracket on the space of observables and provides the nonholonomic dynamics.
\end{abstract}

\section{Introduction}
Nonholonomic dynamics refers to those mechanical systems that are subject to constraints on the velocities (these constraints could be linear or non-linear).

In the Lagrangian picture, a nonholonomic mechanical system is given by a Lagrangian function $L:TQ \to \RR$ defined on the tangent bundle $TQ$ of the configuration manifold $Q$, with nonholonomic constraints provided by a submanifold $D$ of $TQ$. We assume that $\tau_Q(D)=Q$, where $\tau_Q: TQ \to Q$ is the canonical projection to guarantee that we are in presence of purely kinematic constraints. $D$ could be a linear submanifold (in this case, $D$ can be alternatively described by a regular distribution $\Delta$ on $Q$), or nonlinear.

Even if nonholonomic mechanics is an old subject \cite{deLeon2019}, it was in the middle of the nineties that received a decisive boost due to the geometric description by several independent teams: Bloch \emph{et al.}~\cite{Bloch1996}, de León \emph{et al.}~\cite{Ibort1996,deLeon1997a,deLeon1996c,deLeon1996d,deLeon1997b,deLeon1992} and Bates and Śniatycki~\cite{Bates1993}, based on the seminal paper by J.~Koiller in 1992~\cite{Koiller1992}. Another relevant but not so well known work is due to Vershik and Faddeev~\cite{Vershik1972}.

Nowadays, nonholonomic mechanics is a very active area of the so-called Geometric Mechanics.

The geometric description of nonholonomic mechanics uses the symplectic machinery. The idea behind is that there exists an unconstrained system as a background and one can recover the nonholonomic dynamics by projecting, for instance, the unconstrained one. Due to their symplectic backstage, the dynamics is conservative (for linear and ideal constraints).

However, there are other kind of nonholonomic systems that do not fit on the above description. On can imagine, for instance, a nonholonomic system subject additionally to Rayleigh dissipation~\cite{Chaplygin1949,MOSHCHUK1987,Neimark1972}. Another source of examples comes from thermodynamics, treated in~\cite{GayBalmaz2017,GayBalmaz2018,Gay-Balmaz2019} with a variational approach.

Nevertheless, there is a natural geometric description for these systems based on contact geometry.

Contact geometry is, to some extent, an odd-dimensional version of symplectic geometry. In the Lagrangian view, we have a function $L:TQ \times \RR \to \RR$, $L=L(q^i, \dot{q}^i, z)$, where $z$ is a variable indicating friction (or a thermal variable in thermodynamics) and the equations of motion are obtained using the contact 1-form
\begin{equation}
    \eta_L = \dd z - \alpha_L,
\end{equation}
where $\alpha_L$ is the Liouville 1-form on $TQ$ obtained from the regular Lagrangian $L$ and the geometry of $TQ$. The energy of the system is
\begin{equation}
    E_L =  \Delta(L)-L,
\end{equation}
where $\Delta$ is the Liouville vector field and the dynamics is obtained through the equation
\begin{equation}
    \flat_L(\Gamma_L) = \dd E_L - (\Reeb_L (E_L) + E_L) \eta_L,
\end{equation}
where $\Reeb_L$ is the Reeb vector field for the contact form $\eta_L$, and
\begin{equation}
    \flat_L(X) = \contr{X} (\dd \eta_L) + \eta_L(X) \eta_L
\end{equation}
for any vector field $X$ on $TQ \times \RR$. Then, $\Gamma_L$ is a SODE on $TQ \times \RR$, and its solutions (the projection on $Q$ of its integral curves) coincide with the ones of the following equations
\begin{equation}\label{eq:herglotz_intro}
    \pdv{L}{q^i} - \dv{}{t} \pdv{L}{\dot{q}^i} + \pdv{L}{\dot{q}^i} \pdv{L}{z} = 0.
\end{equation}

Amazingly, \cref{eq:herglotz_intro} are the equations obtained by Gustav Herglotz in 1930 \cite{Herglotz1930}, using a new variational principle. For this reason, we will refer from now on to \cref{eq:herglotz_intro} as \emph{Herglotz equations}.

So, our goal in the current paper is just to develop a contact version of  Lagrangian systems with nonholonomic constraints.

The structure of this paper is as follows. \Cref{sec:contact_jacobi} is devoted to describe some introductory material on contact and Jacobi structures; in particular, it is crucial to notice that a contact structure is a particular case of a Jacobi manifold. In \cref{sec:contact_hamiltonian_systems} we discuss the notion of contact Hamiltonian systems and describe the Lagrangian formalism in that context.

\Cref{sec:herglotz_unconstrained,sec:herglotz_constrained} are devoted to analyze the Herglotz principle for nonholonomic systems (a sort of d'Alembert principle in comparison with the well-known Hamilton principle).  The reason to develop this subject is to justify the nonholonomic equations proposed in \cref{sec:herglotz_constrained}.

Then, in \cref{sec:nonholonomic_brackets}, we construct an analog to the symplectic nonholonomic bracket in the contact context. A relevant issue is that this bracket is an almost Jacobi bracket (that is, it does not satisfy the Jacobi identity). This contact nonholonomic bracket transforms the constraints in Casimirs and provides the evolution of observables, as in the unconstrained contact case. In \cref{quizalapenultimaseccion} we introduce the notion of almost Jacobi structure proving that the nonholonomic bracket induces, in fact, an almost Jacobi structure. Then, we prove that this structure is a Jacobi structure if, and only if, the constraints are holonomic.

Finally, we apply our results to a particular example, the \textit{Chaplygin's sledge subject to Rayleigh dissipation}.

\section{Contact manifolds and Jacobi structures}\label{sec:contact_jacobi}
In this section we will introduce the notion of contact structures and Jacobi manifolds. For a detailed study of these structures see~\cite{Vaisman1994,Libermann1987}.

\begin{definition}
A \emph{contact manifold} is a pair $(M,\eta)$, where $M$ is an $(2n+1)$-dimensional manifold and $\eta$ is a $1$-form on $M$ such that $\eta \wedge {(\dd\eta)}^n$ is a volume form. $\eta$ will be called a \emph{contact form}.
\end{definition}

\begin{remark}
Some authors define a contact structure as a smooth distribution $D$ on $M$ such that it is locally generated by contact forms $\eta$ via the kernels $\ker  \eta $. This definition is not equivalent to ours, and is less convenient for our purposes, as explained in~\cite{deLeon2018}.
\end{remark}

As we know, if $(M,\eta)$ is a contact manifold, there exists a unique vector field $\Reeb$, called \emph{Reeb vector field}, such that
\begin{equation}
	\contr{\Reeb}  \dd \eta = 0, \quad  \contr{\Reeb}\, \eta = 1.
\end{equation}

Let us now present the canonical example of contact manifold
\begin{example}\label{48}
Let $\left( x^{i}, y^{j} , z\right)$ be the canonical (global) coordinates of $\mathbb{R}^{2n +1}$. Then,  we may define the following $1-$form $\eta$ on $\mathbb{R}^{2n + 1}$:,  
\begin{equation}
    \eta = \dd z -  y^{i}\dd x^{i}.    
\end{equation}
Hence,
\begin{equation}
    \dd \eta = \dd x^{i} \wedge \dd y^{i}.
\end{equation}
So, the family of vector field $\{X_{i} , Y_{i} \}_{i}$ such that
\begin{equation}
    X_{i} = \frac{\partial}{\partial x^{i}} + y^{i} \frac{\partial}{\partial z}, \quad
    Y_{i} = \frac{\partial}{\partial y^{i}}
\end{equation}
generates the kernel of $\eta$. Furthermore, 
\[\dd \eta \left( X_{i}, X_{j} \right) = \dd \eta ( Y_{i} , Y_{j} ) = 0,
 \quad \dd \eta \left( X_{i} , Y_{j} \right) = \delta^{i}_{j}.\]
Then, $( \mathbb{R}^{2n+1} , \eta )$ is a contact manifold. Notice that, in this case, the Reeb vector field is $\frac{\partial}{\partial z}$.
\end{example}

We may even prove that any contact form locally looks like the contact form defined in example \ref{48}.
\begin{theorem}[Darboux theorem]
Suppose $\eta$ is a contact form on a $2n+1$-dimensional manifold $M$. For each $x \in M$ there are smooth coordinates $(q^i, p_i, z)$ centered at $x$ 
such that
\begin{equation}
	 \eta = \dd z - p_i \, \dd q^i.
\end{equation}
This coordinates will be called \textit{Darboux coordinates}.
\end{theorem} 
Furthermore, in Darboux coordinates the Reeb vector field is expressed by:
\begin{equation}
	\Reeb = \frac{\partial}{\partial z}.
\end{equation}

The contact structure provides a vector bundle isomorphism:
\begin{equation}\label{eq:flat_iso}
    \begin{aligned}
        {\flat} : TM &\to T^* M ,\\
         v &\mapsto \contr{v}  \dd \eta + \eta (v)  \eta.
    \end{aligned}
\end{equation}
We denote also by $\flat:\VecFields(M)\to \Omega^1(M)$ the associated isomorphism of $\mathcal{C}^\infty(M)$-modules. The inverse of $\flat$ will be denoted by $\sharp$. Notice that ${\flat}(\Reeb)=\eta$.

The isomorphism $\flat$ may be defined by contracting a $2-$tensor given by   
\begin{equation}
    \omega = \dd \eta + \eta \otimes \eta.
\end{equation}
So, ${\flat}(X) = \omega(X,\cdot)$ for all $X \in \VecFields(M)$.\\
The form $\eta$ and the Reeb vector field $\Reeb$ provide the following Whitney sum decomposition of the tangent bundle:
\begin{equation}\label{eq:contact_sum_decomposition}
    TM = \ker \eta \oplus \ker \dd\eta.
\end{equation}
Notice that, by counting dimensions, it is easy to realize that $ \ker \dd \eta$ is generated by the Reeb vector field $\Reeb$.\\
We can classify subspaces of the tangent space regarding its relative position to the aforementioned Whitney sum decomposition~\eqref{eq:contact_sum_decomposition}.
\begin{definition}\label{def:subspace_position}
Let $\Delta_x \subseteq T_x M$ be a subspace, where $x\in M$. We say that $\Delta_x$ is
    \begin{enumerate}
        \item \emph{Horizontal} if $\Delta_x \subseteq \ker  \eta_x$.
        \item \emph{Vertical} if $\Reeb_x \in \Delta_x $.
        \item \emph{Oblique} otherwise. By a dimensional counting argument, this is equivalent to $\Delta_x =(\Delta_x \cap \ker  \eta_x)\oplus \gen{\Reeb_x + v_x}$, with $v_x \in \ker  \eta_x \setminus \Delta_x$.
    \end{enumerate}
We say that a distribution $\Delta$ is \emph{horizontal}/\emph{vertical}/\emph{oblique} if $\Delta_x$ is \emph{horizontal}/\emph{vertical}/\emph{oblique}, for every point $x$, respectively.
\end{definition}

 For a distribution $\Delta \subseteq TM$, we define the following notion of complement with respect to $\omega$. Since $\omega$ is neither symmetric nor antisymmetric, left and right complements differ:
\begin{equation}
    \begin{aligned}        
        \orth{\Delta}  &=  
          \set{v \in TM \mid \omega(w,v) = {\flat}(w)(v) = 0,\, \forall w\in \Delta}
          = \ann{({\flat}(\Delta))},\\
        \orthL{\Delta} &= 
          \set{v \in TM\mid \omega(v,w) = 0, \forall w \in \Delta} = \flat^{-1}(\ann{\Delta}).
    \end{aligned}
\end{equation}
These complements have the following relationship
\begin{equation}\label{complements_cancel}
    \orthL{(\orth{\Delta})} = 
    \orth{(\orthL{\Delta})} = \Delta.
\end{equation}
Furthermore, we may interchange sums and intersections, since the annihilator interchanges them and the linear map ${\flat}$ preserves them. Consequently, if $\Delta,\Gamma$ are distributions, we have
\begin{subequations}\label{eq:complement_intersections}
    \begin{align}
        \orth{(\Delta\cap \Gamma)} &= \orth{\Delta} + \orth{\Gamma},\\
        \orth{(\Delta +  \Gamma)} &= \orth{\Delta}  \cap \orth{\Gamma},
    \end{align}
\end{subequations}
and analogously for the left complement.

However, on some important cases, both right and left complements coincide.
\begin{lemma}\label{lem:rl_complement}
    Let $\Delta$ be a distribution on $M$.
    \begin{itemize}
        \item  If $\Delta$ is horizontal, then
        \begin{equation}\label{eq:horizontal_orth}
            \orth{\Delta} = \orthL{\Delta} = \set{v \in TM \mid \dd \eta(v,w) = 0, \, \forall v \in \Delta}.
        \end{equation}
        Hence, $\orth{\Delta}$ is vertical.

        \item  If $\Delta$ is vertical, then
        \begin{equation}
            \orth{\Delta} = \orthL{\Delta} = \set{v \in \ker \eta \mid \dd \eta(v,w) = 0, \, \forall v \in (\Delta \cap \ker \eta)}.
        \end{equation}
        Hence, $\orth{\Delta}$ is horizontal.
    \end{itemize}
\end{lemma}
\begin{proof}
    Assume that $\Delta$ is horizontal. Then for all  $w \in \Delta$,
    \begin{equation}
        \omega(v,w) = \dd \eta(v,w) = - \dd \eta(w,v) = - \omega(w,v),
    \end{equation}
    since $\eta(w) = 0$. So, \cref{eq:horizontal_orth} follows.

    Now, assume $\Delta$ is vertical. By using the Whitney sum decomposition~\eqref{eq:contact_sum_decomposition}, we can write $\Delta = \Delta' \oplus \gen{\Reeb}$, where $\Delta' = \Delta \cap \ker \eta$. Notice that $\orth{\gen{\Reeb}}=\orthL{\gen{\Reeb}} = \ker \eta$. By using~\eqref{eq:complement_intersections}, we deduce
    \begin{align*}
        \orth{\Delta} &= 
        \orth{(\Delta' + \gen{\Reeb})} =
        \orth{\Delta'} \cap \orth{\gen{\Reeb}}  \\ &=
        \orthL{(\Delta' + \gen{\Reeb})} =
        \orthL{\Delta'} \cap \orthL{\gen{\Reeb}}  \\ &=
        \set{v \in TM \mid \dd \eta(v,w) = 0, \, \forall v \in \Delta'} \cap \ker \eta,
    \end{align*}
    as we wanted to show. We have used the first item of the lemma, since $\Delta'$ is horizontal.
\end{proof}

Contact manifolds can be seen as particular cases of the so-called Jacobi manifolds~\cite{Kirillov1976,Lichnerowicz1978}.

\begin{definition}\label{def:jacobi_mfd}
  A \emph{Jacobi manifold} is a triple $(M,\Lambda,E)$, where $\Lambda$ is a bivector field (that is, a skew-symmetric contravariant 2-tensor field) and $E \in \VecFields (M)$ is a vector field, so that the following identities are satisfied:
  \begin{align}
      \lieBr{\Lambda,\Lambda} &= 2 E \wedge \Lambda\\
      \lieD{E} \Lambda &= \lieBr{E,\Lambda} = 0,
  \end{align}
  where $\lieBr{\,\cdot\, , \,\cdot}$ is the Schouten–Nijenhuis bracket.
\end{definition}

The Jacobi structure $(M,\Lambda,E)$ of contact manifold $(M,\eta)$ is given by
\begin{subequations}
  \begin{equation}\label{eq:contact_jacobi}
    \Lambda(\alpha,\beta) = 
    -\dd \eta ({\flat}^{-1} (\alpha), {\flat}^{-1}(\beta)), \quad
    E = - \Reeb.
  \end{equation}
\end{subequations}
  
The Jacobi bivector $\Lambda$ induces a vector bundle morphism between covectors and vectors.
\begin{equation}
    \begin{aligned}
        \sharp_{\Lambda}: T M^* &\to     T M\\
        \alpha &\mapsto \Lambda(\alpha, \cdot ).
    \end{aligned}
\end{equation}

In the case of contact manifolds, the map $\sharp_{\Lambda}$ can be written directly in terms of the contact structure~\cite[Section~3]{deLeon2018} as:
    \begin{equation}
        \sharp_{\Lambda} (\alpha) =
        {{\flat}}^{-1} (\alpha) - \alpha(\Reeb) \Reeb.
    \end{equation}

The Jacobi structure can be characterized in terms of a Lie bracket on the space of functions $\Cont^\infty(M)$, the so-called \emph{Jacobi bracket}. 
\begin{definition}\label{def:jac_bra}
    A \emph{Jacobi bracket} $\jacBr{\cdot,\cdot}: \Cont^\infty(M) \times \Cont^{\infty}(M) \to \Cont^\infty(M)$ on a manifold $M$ is a map that satisfies
    \begin{enumerate}
        \item $(\Cont^\infty(M),\jacBr{\cdot,\cdot})$ is a Lie algebra. That is, $\jacBr{\cdot,\cdot}$ is $\RR$-bilinear, antisymmetric and satisfies the Jacobi identity:
        \begin{equation}
            \jacBr{f,\jacBr{g,h}} + \jacBr{g,\jacBr{h,f}} + \jacBr{h,\jacBr{f,g}} = 0
        \end{equation}
        for arbitrary $f,g,h \in \Cont^\infty(M)$.

        \item It satisfies the following locality condition: for any $f,g \in \Cont^\infty(M)$,
        \begin{equation}
            \supp(\jacBr{f,g}) \subseteq \supp(f) \cap \supp(g),
        \end{equation}
        where $\supp(f)$ is the topological support of $f$, i.e., the closure of the set in which $f$ is non-zero.
    \end{enumerate}

This means that $(\Cont^\infty(M), \jacBr{\cdot,\cdot})$ is a local Lie algebra in the sense of Kirillov~\cite{Kirillov1976}.
\end{definition}

Given a Jacobi manifold $(M,\Lambda, E)$ we can define a Jacobi bracket by setting
\begin{equation}\label{eq:jac_bracket_from_mfd}
    \jacBr{f,g} =\Lambda(\dd f, \dd g) + f E(g) - g E (f).
\end{equation}
In fact, every Jacobi bracket arises in this way.
\begin{theorem}\label{thm:jab_bra_characterization}
    Given a manifold $M$ and a $\RR$-bilinear map $\jacBr{\cdot,\cdot}: \Cont^\infty(M) \times \Cont^{\infty}(M) \to \Cont^\infty(M)$, the following are equivalent.
    \begin{enumerate}
        \item\label{item:jac_bra_local} The map $\jacBr{\cdot,\cdot}$ is a Jacobi bracket.
        
        \item\label{item:jac_bra_leibniz} $(M,\jacBr{\cdot,\cdot})$ is a Lie algebra which satisfies the generalized Leibniz rule
        \begin{equation}\label{eq:mod_leibniz_rule}
            \jacBr{f,gh} = g\jacBr{f,h} + h\jacBr{f,g} +  g h E(h),
        \end{equation}
        where $E$ is a vector field on $M$.

        \item\label{item:jac_bra_mfd} There is a bivector field $\Lambda$ and a vector field $E$ such that $(M,\Lambda,E)$ is a Jacobi manifold and $\jacBr{\cdot,\cdot}$ is given as in \cref{eq:jac_bracket_from_mfd}.
    \end{enumerate}
\end{theorem}
For a proof, see~\cite{deLeon2019}.

\section{Contact Hamiltonian systems}\label{sec:contact_hamiltonian_systems}
Given a Hamiltonian function $H$ on the contact manifold $(M,\eta)$  we define the \emph{Hamiltonian vector field} $X_H$ by the equation
\begin{equation}\label{eq:hamiltonian_vf_contact}
    {\flat} (X_H) = \dd H - (\Reeb (H) + H) \, \eta.
\end{equation}
We call the triple $(M,\eta,H)$ a \emph{contact Hamiltonian system}.
\begin{proposition}\label{prop213123}
Let $H$ be a Hamiltonian function. The following statements are equivalent:
\begin{enumerate}
\item $X_{H}$ is the Hamiltonian vector field of $H$.

\item It satisfies that
\begin{subequations}\label{eqs:ham_vfs_characterization2}
    \begin{align}
        \eta(X_H) &= -H, \\
        \iota_{X_{H}} \dd \eta_{\vert \ker \eta} &= \dd H_{\vert \ker \eta}.\label{eq:ham_vf_conf_contactomorphism2},
    \end{align}
\end{subequations}
\item It satisfies that
\begin{subequations}\label{eqs:ham_vfs_characterization}
    \begin{align}
        \eta(X_H) &= -H, \\
        \lieD{X_H} \eta &= a \eta \label{eq:ham_vf_conf_contactomorphism},
    \end{align}
\end{subequations}
for some function $a$. Notice that this implies that
$a = -\Reeb(H)$
\item Let $(\Lambda,-\Reeb)$ be the Jacobi structure induced by the contact form $\eta$ (see \cref{eq:contact_jacobi}). Then,
\begin{equation}
     X_H = \sharp_{\Lambda}(\dd H) - H \Reeb.
\end{equation}
\end{enumerate}
\end{proposition}

By Cartan's formula,  the following identity is also satisfied,
\begin{subequations}
    \begin{align}
        \contr{X_H} \dd \eta &= \dd H - \Reeb(H) \eta.
    \end{align}
\end{subequations}

Contrary to the symplectic case, the energy and the phase space volume are not conserved (see~\cite{deLeon2018}).
\begin{proposition}[Energy and volume dissipation]\label{prop:energy_disipation}
    Given a Hamiltonian system $(M,\eta,H)$, the energy is not preserved along the flow of $X_H$. Indeed,
    \begin{equation}
        \lieD{X_H} H = - \Reeb(H) H.
    \end{equation}
    
    The contact form is also not preserved:
    \begin{equation}
        \lieD{X_H} \eta = - \Reeb(H) \eta.
    \end{equation}

    However, if $H$ does not vanish, the modified contact form
    \begin{equation}
        \tilde{\eta} = \frac{\eta}{H}
    \end{equation}
    is preserved. That is,
    \begin{equation}
        \lieD{X_H}\tilde\eta= 0.
    \end{equation}
    Moreover, $-X_H$ is the Reeb vector field of $\tilde{\eta}$.
    
   The contact volume element $\Omega = \eta\wedge(\dd \eta)^n$ is not preserved. In fact,
    \begin{equation}
        \lieD{X_H} \Omega=
         - (n+1) \Reeb (H) \Omega.
    \end{equation}

    However, if $H$ does not vanish, the following modified volume element is preserved:
    \begin{equation}
        \tilde\Omega = {H}^{-(n+1)}  \Omega = \tilde{\eta}\wedge (\dd \tilde{\eta})^n,
    \end{equation}
    that is,
    \begin{equation}
        \lieD{X_H}\tilde\Omega = 0.
    \end{equation}
\end{proposition}

An important case of contact manifold is the manifold $TQ \times \mathbb{R}$ where $Q$ is an $n-$dimensional manifold. 
If $L:TQ\times \RR \to \RR$ is \textit{regular Lagrangian function}, that is, its Hessian matrix with respect to the velocities
\begin{equation}\label{eq:velocity_Hessian}
	W_{ij}=\left( \frac{\partial^2 L}{\partial \dot{q}^i \partial \dot{q}^j} \right)
\end{equation}
is regular, then $(TQ \times \RR, \eta_L)$ is a contact manifold. Here $(q^i, \dot{q}^i,z)$ are the natural coordinates on $TQ \times \RR$ induced by coordinates $\left( q^{i} \right)$ on $Q$. The contact form $\eta_{L}$  given by
\begin{equation}
    \eta_L = \dd z - \alpha_L,
\end{equation}
where
\begin{align}
    \alpha_L &= S^* (\dd L) = \pdv{L}{\dot q^i} \dd {q}^i,
\end{align}
 and $S$ is the canonical vertical endomorphism on $TQ$ extended in the natural way to $TQ \times \RR$ (see \cite{deLeon2011} for an intrinsic definition). The form $\eta_{L}$ is called the \emph{contact Lagrangian form}. Locally, we have that $\eta_{L} = \dd z -  \pdv{L}{\dot{q}^i} {\dd q}^{i}$ and, hence

 \begin{equation}\label{puesestaecuacionsera}
    \dd \eta_{L} = \dfrac{\partial^{2}L}{\partial\dot{q}^i \partial q^{k}}{\dd q}^{i} \wedge {\dd q}^{k} + \dfrac{\partial^{2}L}{\partial\dot{q}^i \partial \dot{q}^{k}}{\dd q}^{i} \wedge \dd\dot{q}^{k} + \dfrac{\partial^{2}L}{\partial\dot{q}^i \partial z}{\dd q}^{i} \wedge \dd z.
 \end{equation}
The energy of the system is defined by
 \begin{equation}
    E_L = \Delta(L) - L = q^i \pdv{L}{q^i} - L,
 \end{equation}
 where $\Delta = q^i \pdv{}{q^i}$ is the Liouville vector field on $TQ$ extended to $TQ\times \RR$ in the natural way .

 Hence, $(TQ\times \RR, \eta_L, E_L)$ is a contact Hamiltonian system.

The Reeb vector field will be denoted by $\Reeb_L$ and it is given in bundle coordinates by
\begin{equation}\label{50}
    \Reeb_L = \frac{\partial}{\partial z} - 
    W^{ij} \frac{\partial^2 L}{\partial \dot{q}^i \partial z} \pdv{}{\dot{q}^j},
\end{equation}
where $(W^{ij})$ is the inverse of the Hessian matrix with respect to the velocities~\eqref{eq:velocity_Hessian}.

The Hamiltonian vector field of the energy will be denoted by $\Gamma_L = X_{E_L}$, hence
\begin{equation}\label{51}
    \flat_L(\Gamma_L) = \dd E_L - (E_{L} + \Reeb(E_L))\eta_L,
\end{equation}
where $\flat_L(v) = \contr{v} \dd \eta_L + \eta_L (v) \eta_L$ is the isomorphism defined in \cref{eq:flat_iso} for this particular contact structure.

\begin{definition}
Let us consider a regular Lagrangian $L: TQ \times \mathbb{R} \rightarrow \mathbb{R}$. Then, a vector field $X$ on $TQ \times \mathbb{R}$ is called a \textit{SODE} if it satisfies that
\begin{equation}\label{55}
S \left( X \right) = \Delta
\end{equation}
\end{definition}
Let $\left( q^{i} \right) $ be a local system of coordinates on $Q$. Then, a vector field on $TQ \times \mathbb{R}$ is (locally) expressed as follows
$$X\left( q^{i} \right) \dfrac{\partial}{\partial q^{i}}  + X\left( \dot{q}^{i} \right) \dfrac{\partial}{\partial \dot{q}^{i}}  + X\left( z \right) \dfrac{\partial}{\partial z}.$$
So, Eq. (\ref{55}) reduces to
$$ X \left( \dot{q}^{i} \right) = \dot{q}^{i}, \ \forall i.$$
So, it may be easily checked that the vector field $\Gamma_L$ is a SODE (\cref{49}).

\begin{proposition}
Let $X$ be a vector field on $TQ \times \mathbb{R}$. $X$ is a SODE if, and only if, any integral curve of $X$ is written (locally) as $\left( \xi , \dot{\xi} , s \right)$ for some (local) path $\xi$ on $Q$. 

\end{proposition}

\section{The Herglotz variational principle}\label{sec:herglotz_unconstrained}
The Herglotz principle was introduced by G. Herglotz \cite{Herglotz1930}, and has been rediscovered by B. Georgieva \emph{et al.}~\cite{Georgieva2003,Georgieva2011}. More recently, A. Bravetti \emph{et al.}~\cite{Bravetti2019} has connected this principle with contact Hamiltonian dynamics (see also \cite{deLeon2019} for a more recent approach).
Let us now present this principle to end up deriving the Herglotz equations. Consider a Lagrangian function $L:TQ \times \RR \to \RR$ and fix two points $q_1,q_2 \in Q$ and an interval $[a,b] \subset \RR$. We denote by $\Omega(q_1,q_1, [a,b]) \subseteq(\Cont^\infty([a,b]\to Q)$ the space of smooth curves $\xi$ such that $\xi(a)=q_1$ and $\xi(b)=q_2$. This space has the structure of an infinite dimensional smooth manifold whose tangent space at $\xi$ is given by the set of vector fields over $\xi$ that vanish at the endpoints~\cite[Proposition~3.8.2]{Abraham1978}, that is,
\begin{equation}
\begin{aligned}
        T_\xi \Omega(q_1,q_2, [a,b]) =  \set{&
            v_\xi \in \Cont^\infty([a,b] \to TQ) \mid \\& 
            \tau_Q \circ v_\xi = \xi, \,v_\xi(a)=0, \, v_\xi(b)=0 
            }.
\end{aligned}
\end{equation}
The elements of $T_{\xi} \Omega(q_1,q_2, [a,b]) $ will be called \textit{infinitesimal variations} of the curve $\xi$.
We will consider the following maps. Let 
\begin{equation}
    \mathcal{Z}:\Cont^\infty ([a,b] \to Q) \to \Cont^\infty ([a,b] \to \RR)
\end{equation}
 be the operator that assigns to each curve $\xi$ the curve $\mathcal{Z}(\xi)$ that solves the following ordinary differential equation (ODE):
\begin{equation}\label{1}
    \dv{\mathcal{Z}(\xi)(t)}{t} = L(\xi(t), \dot \xi(t), \mathcal{Z}(\xi)(t)), \quad \mathcal{Z}(\xi)(a)= 0.
\end{equation}

Now we define the \emph{action functional} as the map which assigns to each curve the solution to the previous ODE evaluated at the endpoint, namely,
\begin{equation}\label{2}
    \begin{aligned}
        \mathcal{A}: \Omega(q_1,q_2, [a,b]) &\to \RR,\\
        \xi &\mapsto \mathcal{Z}(\xi)(b).
    \end{aligned}
\end{equation}
We will say that a path $\xi \in \Omega \left( q_1,q_2, Q \right)$ satisfies the \emph{Herglotz variational principle} if it is a critical point of $\mathcal{A}$, i.e.,
\begin{equation}\label{3}
T_{\xi}\mathcal{A} = 0
\end{equation}
As we have shown in~\cite[Thm.~2]{deLeon2019}, the functional derivative of $\mathcal{A}$ is given by
\begin{equation}\label{4}
    T_{\xi} \mathcal{A}(v) = \frac{1}{\sigma(b)} \int_a^b v^i(t) \sigma(t)
    \left(\pdv{L}{q^i} - \dv{}{t} \pdv{L}{\dot{q}^i} + \pdv{L}{\dot{q}^i} \pdv{L}{z} \right) \dd t,
\end{equation}
where $\xi \in \Omega(q_1,q_1, [a,b])$, $v \in T_\xi \Omega(q_1,q_1, [a,b])$ and 
\begin{equation}
     \sigma(t) = \exp \left({-\int_a^t \pdv{L}{z} \dd \tau}\right).
\end{equation}

The critical points of this functional are those curves $\xi$ along which the following equations are satisfied:
\begin{equation}\label{52}
   \pdv{L}{q^i} - \dv{}{t} \pdv{L}{\dot{q}^i} + \pdv{L}{\dot{q}^i} \pdv{L}{z} = 0.
\end{equation}
These equations are called \textit{Herglotz equations}.

As it is proved in~\cite[Section~3]{deLeon2019}, for a regular Lagrangian $L$, Herglotz equations are equivalent to those obtained in a geometric way using the contact structure $\eta_{L}$ induced by the Lagrangian.

\begin{theorem}\label{49}
Assume that $L$ is regular. Let $\Gamma_{L}$ be the Hamiltonian vector field on $TQ \times \RR$ associated to the energy (see Eq. (\ref{51})), i.e.,
$$\flat_L(\Gamma_L) = \dd E_L - (\Reeb(E_L)+E_L)\eta_L.$$
Then, 
\begin{enumerate}
\item $\Gamma_{L}$ is a SODE on $TQ \times \RR$.

\item The integral curves of $\Gamma_{L}$ are solutions of the Herglotz equations~\eqref{52}.
\end{enumerate}
\begin{proof}
If we write \cref{51} in local coordinates we have that
\begin{equation}\label{53}
\Gamma_{L}= \dot{q}^{i} \pdv{}{q^{i}} + b^{i}\pdv{}{\dot{q}^{i}} + L \pdv{}{z},
\end{equation}
and the local functions $b^{i}$ satisfy
\begin{equation}\label{54}
b^{k} \dfrac{\partial^{2}L}{\partial\dot{q}^k \partial \dot{q}^{i}} + \dot{q}^{k}\dfrac{\partial^{2}L}{\partial  q^k \partial \dot{q}^{i}} + L \dfrac{\partial^{2}L}{\partial z \partial \dot{q}^{i}} - \dfrac{\partial  L }{\partial q^{i}} =\dfrac{\partial  L }{\partial \dot{q}^{i}}\dfrac{\partial  L }{\partial z}.
\end{equation}
The result follows from this expression.

\end{proof}
\end{theorem}
Sometimes, these equations are called \textit{generalized Euler-Lagrange equations}.

\section{Herglotz principle with constraints}\label{sec:herglotz_constrained}

Consider now that the system is restricted to certain (linear) constraints on the velocities modelized by a regular distribution $\Delta$ on the configuration manifold $Q$ of codimension $k$. Then, $\Delta$ may be locally described in terms of independent linear constraint functions $\{\Phi^a \}_{a=1, \dots , k}$ in the following way
\begin{equation} \label{5}
\Delta=\left\{v\in TQ \mid \Phi^a \left( v \right)=0\right\}.
\end{equation}
Notice that, due to the linearity, the constraint functions $\Phi^{a}$ may be considered as 1-forms $\Phi^{a}: Q \rightarrow T^{*}Q$ on $Q$. Without danger of confusion, we will also denote by $\Phi^{a}$ to the 1-form version of the constraint $\Phi^{a}$ This means that
$$\Phi^a = \Phi^a_i (q) \dot{q}^i.$$
Let $L: TQ \times \mathbb{R} \rightarrow \mathbb{R}$ be the Lagrangian function. One may then define the \emph{Herglotz variational principle with constraints}, that is, we want to find the paths $\xi \in \Omega(q_1,q_1, [a,b])$ satisfying the constraints such that $T_\xi \mathcal{A}(v) = 0$ for all infinitesimal variation $v$ which is tangent to the constraints $\Delta$. More precisely, we define the set
\begin{equation}
    \Omega(q_1,q_1, [a,b])_{\xi}^{\Delta} = \set{
        v \in T_\xi \Omega(q_1,q_1, [a,b]) \mid v(t) \in \Delta_{\xi(t)} \text{ for all } t \in [a,b]}.
\end{equation}
Then, $\xi$ satisfies the Herglotz variational principle with constraints if, and only if,
\begin{enumerate}
\item $T_{\xi}\mathcal{A}_{\vert \Omega(q_1,q_1, [a,b])_{\xi}^{\Delta}} = 0.$

\item $\dot{\xi} \left( t \right) \in \Delta_{\xi(t)} \text{ for all } t \in [a,b].$
\end{enumerate}

\begin{definition}
\rm
A \textit{constraint Lagrangian system} is given by a pair $\left( L , \Delta \right)$ where $L : TQ \times \RR \rightarrow \RR$ is a regular Lagrangian and $\Delta$ is a regular distribution on $Q$. The constraints are said to be \textit{semiholonomic} if $\Delta$ is involutive and \textit{non-holonomic} otherwise.
\end{definition}

Using Eq. (\ref{4}) one may easily prove the following characterization of the Herglotz variational principle with constraints.
\begin{theorem}
A path $\xi \in \Omega(q_1,q_1, [a,b])$ satisfies the Herglotz variational principle with constraints if, and only if,
\begin{equation}
    \begin{cases}  
\pdv{L}{q^i} - \dv{}{t} \pdv{L}{\dot{q}^i} + \pdv{L}{\dot{q}^i} \pdv{L}{z}  \in \ann{\Delta}_{\xi(t)}
\\
\dot{\xi} \in \Delta
\end{cases}
\end{equation}
where $\ann{\Delta} = \set {a \in T^*Q \mid a(u)=0 \text{ for all } u \in \Delta }$ is the annihilator of $\Delta$.
\end{theorem}
Taking into account Eq. (\ref{5}), we have that $\ann{\Delta}$ is (locally) generated by the one-forms $\Phi^a $. Then $\xi$ satisfies the Herglotz variational principle with constraints if, and only if, it satisfies the following equations
\begin{equation}
    \begin{cases}\label{eq:nonholonomic_herglotz_eqs_coords}
        \dv{}{t} \pdv{L}{\dot{q}^i} - \pdv{L}{q^i} - \pdv{L}{\dot{q}^i} \pdv{L}{z} = \lambda_a\Phi^a_i \\
        \Phi^a (\dot{\xi}(t)) = 0.
    \end{cases}
\end{equation} 
for some Lagrange multipliers $\lambda_i(q^i)$ and where $\Phi^a= \Phi^a_i\dd q^i$.\\
From now on, Eqs. (\ref{eq:nonholonomic_herglotz_eqs_coords}) will be called \emph{constraint Herglotz equations}.

We will now present a geometric charaterization of the Herglotz equations similar to \cref{49}. In order to do this, we will consider a distribution $\Delta^{l}$ on $TQ\times \mathbb{R}$ induced by $\Delta$ such that its annihilator is given by

\begin{equation}
    \ann{\Delta^{l}} = \left( \tau_{Q} \circ \pr_{TQ \times \mathbb{R}}\right)^{*} \Delta^{0},
\end{equation}
where $\tau_{Q} : TQ \rightarrow Q$ is the canonical projection and $ \pr_{TQ \times \mathbb{R}}: TQ \times \RR \rightarrow TQ$ is the projection on the first component. In fact, we may prove that
\begin{equation}
    \Delta^l = S^{*} \left( \ann{T \left(\Delta \times \mathbb{R}\right)}\right).
\end{equation}
Hence, $\ann{\Delta^{l}}$ is generated by the 1-forms on $TQ \times \mathbb{R}$ given by
\begin{equation}\label{eq:delta_l_constraints}
\tilde{\Phi}^{a}=\Phi^{a}_{i}{\dd q}^{i}
\end{equation}
Then, we have the following result.
\begin{theorem}\label{18}
Assume that $L$ is regular. Let $X$ be a vector field on $TQ \times \RR$ satisfying the equation
\begin{equation}
    \begin{cases}\label{6}
        \flat_L\left(X\right) - \dd E_L + \left(E_L + \Reeb_L\left(E_L\right)\right)\eta_L \in \ann{\Delta^{l}} \\
        X_{\vert \Delta \times \RR} \in  \mathfrak{X}\left(\Delta \times \RR \right).
    \end{cases},
\end{equation}
Then, 
\begin{itemize}
\item[(1)] $X$ is a SODE on $TQ \times \RR$.

\item[(2)] The integral curves of $X$ are solutions of the constraint Herglotz equations (\ref{eq:nonholonomic_herglotz_eqs_coords}).
\end{itemize}
\begin{proof}
To prove this theorem we will take advantage of the calculations done for the Herglotz equations (\ref{52}). So, consider a vector field $Y$ on $TQ \times \RR$ such that
\begin{equation}\label{12}
 \flat_L \left( Y \right) \in \ann{\Delta^{l}}.
 \end{equation}
Then, there exist some (local) functions $\lambda_{a}$ such that
$$ \flat_L \left( Y \right)  =  \lambda_a\Phi^a_i {\dd q}^{i}.$$
Observe that $\eta_{L} = \dd z -  \pdv{L}{\dot{q}^i} {\dd q}^{i}$ and, hence (\ref{puesestaecuacionsera})
$$ \dd \eta_{L} = \dfrac{\partial^{2}L}{\partial\dot{q}^i \partial q^{k}}{\dd q}^{i} \wedge {\dd q}^{k} + \dfrac{\partial^{2}L}{\partial\dot{q}^i \partial \dot{q}^{k}}{\dd q}^{i} \wedge \dd\dot{q}^{k} + \dfrac{\partial^{2}L}{\partial\dot{q}^i \partial z}{\dd q}^{i} \wedge \dd z.$$
Thus,
\begin{itemize}
\item[(i)] $\iota_{\pdv{}{q^i}}\dd \eta_{L} = \left[\dfrac{\partial^{2}L}{\partial\dot{q}^i \partial q^{k}} - \dfrac{\partial^{2}L}{\partial\dot{q}^k \partial q^{i}}\right] {\dd q}^{k} + \dfrac{\partial^{2}L}{\partial\dot{q}^i \partial \dot{q}^{k}}\dd \dot{q}^{k} + \dfrac{\partial^{2}L}{\partial\dot{q}^i \partial z}\dd z.$
\item[(ii)] $\iota_{\pdv{}{\dot{q}^i}}\dd \eta_{L} = - \dfrac{\partial^{2}L}{\partial\dot{q}^i \partial \dot{q}^{k}}{\dd q}^{k}.$
\item[(iii)] $\iota_{\pdv{}{z}}\dd \eta_{L} = - \dfrac{\partial^{2}L}{\partial\dot{q}^k \partial z}{\dd q}^{k}.$
\end{itemize}
Then, 
\begin{eqnarray*}
0 & = & \{ \flat_L \left( Y \right) \} \left( \pdv{}{\dot{q}^i}\right) = \iota_{Y} \dd \eta_{L} \left( \pdv{}{\dot{q}^i}\right)  + \eta_{L} \left( Y \right)  \eta_{L} \left( \pdv{}{\dot{q}^i}\right)\\
&=& \iota_{\pdv{}{\dot{q}^i}}\dd \eta_{L}  \left( Y \right) = \dfrac{\partial^{2}L}{\partial\dot{q}^i \partial \dot{q}^{k}} Y \left(q^{k} \right)
\end{eqnarray*}
Therefore, due to the regularity of $L$ we have that $Y \left(q^{k} \right) = 0$ for all $k$. Analogously, contracting by $\pdv{}{z}$, we obtain that $Y \left( z \right) = 0$. This proves that
$$ \eta_{L}\left( Y \right) = 0.$$
Finally
\begin{eqnarray*}
\lambda_a\Phi^a_i & = & \{ \flat_L \left( Y \right) \} \left( \pdv{}{q^i}\right) = \iota_{Y} \dd \eta_{L} \left( \pdv{}{q^i}\right)  + \eta_{L} \left( Y \right)  \eta_{L} \left( \pdv{}{q^i}\right)\\
&=&  -  \dfrac{\partial^{2}L}{\partial\dot{q}^i \partial \dot{q}^{k}} Y \left( \dot{q}^{k}\right)
\end{eqnarray*}
So,
\begin{equation}\label{13}
\dfrac{\partial^{2}L}{\partial\dot{q}^i \partial \dot{q}^{k}} Y \left( \dot{q}^{k}\right) = - \lambda_a\Phi^a_i
\end{equation}

Next, consider a vector field $X$ satisfying Eq. (\ref{6}). Notice that $ \flat_L(X) - \dd E_L + \left(E_L + \Reeb_L\left(E_L\right)\right)\eta_L \in \ann{\Delta^{l}}$ if, and only if, $\ \flat_L\left( \Gamma_{L}- X\right) \in \ann{\Delta^{l}}$ where $\Gamma_{L}$ is the solution of the Herglotz equations~\eqref{52}. Then, denoting $Y  = \Gamma_{L}- X$, we have that
\begin{equation}\label{38}
X = \dot{q}^{i} \pdv{}{q^{i}} + \left[\Gamma_{L}\left( \dot{q}^{i}\right) - Y \left( \dot{q}^{i}\right) \right] \pdv{}{\dot{q}^{i}} + L \pdv{}{z}.
\end{equation}
This proves that $X$ is a SODE. So, by using \cref{54} we have that
$$\Gamma_{L}\left( \dot{q}^{k}\right) \dfrac{\partial^{2}L}{\partial\dot{q}^k \partial \dot{q}^{i}} =  - \dot{q}^{k}\dfrac{\partial^{2}L}{\partial  q^k \partial \dot{q}^{i}} - L \dfrac{\partial^{2}L}{\partial z \partial \dot{q}^{i}} + \dfrac{\partial  L }{\partial q^{i}} +\dfrac{\partial  L }{\partial \dot{q}^{i}}\dfrac{\partial  L }{\partial z}.$$
Then, using \cref{13}, we obtain
\begin{eqnarray*}
X\left( \dot{q}^{k}\right) \dfrac{\partial^{2}L}{\partial\dot{q}^k \partial \dot{q}^{i}} &=&  - \dot{q}^{k}\dfrac{\partial^{2}L}{\partial  q^k \partial \dot{q}^{i}} - L \dfrac{\partial^{2}L}{\partial z \partial \dot{q}^{i}} + \dfrac{\partial  L }{\partial q^{i}} +\dfrac{\partial  L }{\partial \dot{q}^{i}}\dfrac{\partial  L }{\partial z} - Y \left( \dot{q}^{k}\right) \dfrac{\partial^{2}L}{\partial\dot{q}^k \partial \dot{q}^{i}}\\
&=& - \dot{q}^{k}\dfrac{\partial^{2}L}{\partial  q^k \partial \dot{q}^{i}} - L \dfrac{\partial^{2}L}{\partial z \partial \dot{q}^{i}} + \dfrac{\partial  L }{\partial q^{i}} +\dfrac{\partial  L }{\partial \dot{q}^{i}}\dfrac{\partial  L }{\partial z} +\lambda_a\Phi^a_i
\end{eqnarray*}
Notice that any integral curve of $X$ may be written (locally) as $\left( \xi , \dot{\xi} , \mathcal{Z}\right)$ with $ \dot{\mathcal{Z}} = L \left( \xi , \dot{\xi} , \mathcal{Z}\right)$ (see Eq. (\ref{38})). Then, we have that
$$\ddot{q}^{k} \dfrac{\partial^{2}L}{\partial\dot{q}^k \partial \dot{q}^{i}} = - \dot{q}^{k}\dfrac{\partial^{2}L}{\partial  q^k \partial \dot{q}^{i}} - \dot{\mathcal{Z}} \dfrac{\partial^{2}L}{\partial z \partial \dot{q}^{i}} + \dfrac{\partial  L }{\partial q^{i}} +\dfrac{\partial  L }{\partial \dot{q}^{i}}\dfrac{\partial  L }{\partial z} +\lambda_a\Phi^a_i .$$
Equivalently,
$$\dv{}{t} \pdv{L}{\dot{q}^i} - \pdv{L}{q^i} - \pdv{L}{\dot{q}^i} \pdv{L}{z} = \lambda_a\Phi^a_i.$$
Let us now study the second equation in (\ref{6}), that is, $X_{\vert \Delta \times \RR} \in  \frak X \left(\Delta \times \RR \right)$. Let $\left( \xi , \dot{\xi} , \mathcal{Z}\right)$ be an integral curve of $X_{\vert \Delta \times \RR}$. Then, by the condition of tangency we deduce that  $\left( \xi , \dot{\xi} \right) \subseteq \Delta$, i.e.,
$$\Phi^{a} \left( \xi \left( t \right) , \dot{\xi} \left( t \right) \right) = 0, \ \forall a.$$
\end{proof}
\end{theorem}

Therefore, Eq. (\ref{6}) provides the correct nonholonomic dynamics in the context of contact geometry. In the case of existence and uniqueness, the particular solution to Eq. (\ref{6}) will be denoted by $\Gamma_{L , \Delta}$. We will now investigate the existence and uniqueness of the solutions.\\
\begin{remark}[The distribution $\Delta^l$]\label{rem:DeltaL}
    From the coordinate expression of the constraints $\tilde{\Phi}^a$ defining $\Delta^l$ (\cref{eq:delta_l_constraints}), one can see that $\Reeb_{L}(\tilde{\Phi}^a)=0$, hence $\Delta^l$ is vertical in the sense of~\cref{def:subspace_position}.
\end{remark}
\begin{remark}
Notice that $T \left( \Delta \times \RR \right)$ may be considered as a distribution of $TQ \times \RR$ along the submanifold $ \Delta \times \RR $. Then, it is easy to show that the annihilator of the distribution $T \left( \Delta \times \RR \right)$ is given by $\pr_{\Delta \times \mathbb{R}}^{*} \ann{\left(T \Delta \right)}$ where $\pr_{\Delta \times \mathbb{R}}: \Delta \times \mathbb{R} \rightarrow \Delta$ denotes the projection on the first component. In fact, let $ \left( X , f \right)$ be a vector field on $ TQ \times \RR $, that is, for all $\left( v_{q} , z \right) \in T_{q}Q \times \RR$ we have that
\[ X \left( v_{q} , z \right) \in T_{v_{q}}\left( TQ \right) ; \quad f \left( v_{q} , z \right) \in T_{z}\RR \cong \RR.\]
Then, for each $\left( v_{q} , z \right) \in \Delta \times \RR$,  $\left( X , f \right) \left( v_{q} , z \right)$ is tangent to $\Delta \times \RR$ at $\left( v_{q} , z \right) $ if, and only if,
\begin{equation}\label{16}
\dd \Phi^{a}_{\vert v_{q} } \left(  X \left( v_{q} , z \right) \right) = 0 , \ \forall a.
\end{equation}
Denoting $\overline{\Phi}^{a} = \Phi^{a}\circ \pr_{\Delta \times \mathbb{R}}$, we may express Eq. (\ref{16}) as follows
\begin{equation}\label{23}
 Z \left( \overline{\Phi}^{a}\right) = 0, \ \forall a
\end{equation}
where $ Z =\left( X , f \right)$. It is important to notice that, being $\Delta = \left(\Phi^{a}\right)^{-1} \left( 0 \right)$, it satisfies that
\[T_{\left( v_{q} , z \right)} \left( \Delta \times \RR \right) = \ker \left( T_{v_{q}} \left(\Phi^{a}\right) \right) \times \RR.\]
\end{remark}

Let $\mathcal{S}$ be the distribution on $TQ \times \RR$ defined by $\sharp_{L} \left( \ann{\Delta^{l}}\right)$ where $\sharp_{L} = \flat_{L}^{-1}$.

In order to find a (local) basis of sections of $\mathcal{S}$, we will consider the 1-forms $\tilde{\Phi}^{a} $ generating $\ann{\Delta^{l}}$. For each $a$, $Z_{a}$ will be the local vector field on $TQ \times \RR$ satisfying
\begin{equation}\label{14}
\flat_L \left( Z_{a} \right) = \tilde{\Phi}^{a}.
\end{equation}
Then, $\mathcal{S}$ is  obviously (locally) generated by the vector fields $Z_{a}$ and $\mathcal{S} \subseteq \Delta^{l}$.

By using the proof of the theorem \ref{18} we have that
\begin{equation}
Z_{a}\left( q^{i} \right) = Z_{a} \left( z \right) = 0, \ \ \ \dfrac{\partial^{2}L}{\partial\dot{q}^i \partial \dot{q}^{k}} Z_{a} \left( \dot{q}^{k}\right) = - \Phi^a_i
\end{equation}
Then,
\begin{equation}\label{15}
Z_{a} = -W^{ik} \Phi^{a}_{k}\dfrac{\partial}{\partial \dot{q}^{i}},
\end{equation}
where $\left( W^{ik} \right)$ is the inverse of the Hessian matrix $\left( W_{ik}\right) = \left(\dfrac{\partial^{2}L}{\partial\dot{q}^i \partial \dot{q}^{k}}\right)$.
Notice that, taking into account that $\ann{\Delta^{l}}$ is generated by the 1-forms on $TQ \times \mathbb{R}$ given by $\tilde{\Phi}^{a}=\Phi^{a}_{i}{\dd q}^{i}$, it follows that
\begin{equation}\label{33}
\mathcal{S} \subseteq \Delta^{l}
\end{equation}

\begin{remark}[The distribution $\mathcal{S}$]\label{rem:S}
    Notice, that, since $\Delta^l$ is vertical (\cref{rem:DeltaL}), by \cref{lem:rl_complement} $\mathcal{S} = \orthL{(\Delta^l)} = \orth{(\Delta^l)}$ and $\mathcal{S}$ is horizontal. Hence $\eta_{L}(\mathcal{S})=0$.         
\end{remark}

Assume now that there exist two solutions $X$ and $Y$ of Eq. (\ref{6}). Then, by construction we have that $X-Y$ is tangent $T \left( \Delta \times \RR \right)$. On the other hand,
\[\flat_{L} \left (X - Y \right) = \flat_{L} \left (X - \Gamma_{L}\right) + \flat_{L} \left (\Gamma_{L} - Y \right) \in \ann{\Delta^{l}}.\]
Then, $X-Y$ is also tangent to $\mathcal{S}$. Thus, we may prove the folloeing result:
\begin{proposition}\label{Yoootraproposiciondelcarajomas}
The uniqueness of solutions of \cref{6} is equivalent to 
$$\mathcal{S} \cap T \left( \Delta \times \RR \right)= \{0\}.$$
\begin{proof}
Let $X$ be a solution of \cref{6}. Then, $X + \Gamma$ is a new solution of \cref{6} for any $\Gamma \in \mathcal{S} \cap T \left( \Delta \times \RR \right)$.
\end{proof}
\end{proposition}

if the intersection $\mathcal{S} \cap T \left( \Delta \times \RR \right)$ were zero, we would be able to ensure the uniqueness of solutions.\\
Let $X =X^{b}Z_{b} $ be a vector field on $\Delta \times \RR$ tangent to $\mathcal{S}$. Hence, by Eq. (\ref{23}), we have that
\[ X^{b} \dd\overline{\Phi}^{a}\left( Z_{b} \right) = 0.\]
Equivalently,
\[ X^{b}W^{ik}\Phi^{b}_{k}\Phi^{a}_{i} =  0, \ \forall a.\]
Define the (local) matrix $\mathcal{C}$ with coeficient

\begin{equation}\label{lamatrizcdecarajo}
\mathcal{C}_{ab} = - W^{ik}\Phi^{b}_{k}\Phi^{a}_{i} = \dd\Phi^{b}\left( Z_{a}\right)
\end{equation}
Then, it is easy to prove that (locally) the regularity of $\mathcal{C}$ is equivalent $\mathcal{S} \cap T \left( \Delta \times \RR \right)= \{0\}$.\\
One can easily verify that if the Hessian matrix $\left( W_{ik} \right)$ is positive or negative definite this condition is satisfied.\\\\
From now on we will assume that the Hessian matrix $\left( W_{ik} \right)$ is positive (or negative) definite.\\

\begin{remark}
\rm
In general, we may only assume that the matrices $\mathcal{C}$ are regular. However, for applications, in the relevant cases the Hessian matrix $\left( W_{ik} \right)$ is positive definite. In particular, if the Lagrangian $L$ is natural, that is, $L = T+ V\left( q, z \right)$, where $T$ is the \textit{kinetic energy} of a Riemannian metric $g$ on $Q$ and $V$ is a \textit{potential energy}, then the Lagrangian $L$ will be positive definitive.\\

\end{remark}
\noindent{Notice that, for each $\left( v_{q} , z \right) \in \Delta \times \RR$ we have that}
\begin{itemize}
\item $\dim \left( S_{ \vert \left( v_{q} , z \right)} \right) = k$

\item $\dim \left(T_{ \left( v_{q} , z \right)} \left( \Delta \times \RR \right)\right) = 2n +1 -k $
\end{itemize}
So, the condition of being positive (or negative) definite not only implies that $\mathcal{S} \cap T \left( \Delta \times \RR \right)= \{0\}$ but also we have
\begin{equation}\label{41}
\mathcal{S} \oplus T \left( \Delta \times \RR \right) = T_{\Delta \times \RR}\left( TQ \times \RR \right),
\end{equation}
where $T_{\Delta \times \RR}\left( TQ \times \RR \right)$ consists of the tangent vectors of $TQ \times \RR $ at points of $\Delta \times \RR $.\\
Thus, the uniqueness condition will imply the existence of solutions of \cref{6}. In fact, we will also be able to obtain the solutions of Eq. (\ref{6}) in a very simple way. In fact, let us consider the two projectors
\begin{subequations}\label{projectors2}
\begin{align}
\mathcal{P}: T_{\Delta \times \RR}\left( TQ \times \RR \right) &\rightarrow T \left(  \Delta \times \RR \right),\\
\mathcal{Q} :  T_{\Delta \times \RR}\left( TQ \times \RR \right) &\rightarrow \mathcal{S}.
\end{align}
\end{subequations}

Consider $X = \mathcal{P} \left( {\Gamma_{L}}_{\vert \Delta \times\RR} \right)$. Then, by definition $X \in  \frak X \left(\Delta \times \RR \right)$. On the other hand, at the points in $\Delta \times\RR$ we have
\begin{align*}
 &\flat_{L} \left(X \right) - \dd E_L + \left(E_L + \Reeb_L\left(E_L\right)\right)\eta_L \\ =&
\flat_{L} \left(\Gamma_{L} - \mathcal{Q}\left( \Gamma_{L} \right) \right) - \dd E_L + \left(E_L + \Reeb_L\left(E_L\right)\right)\eta_L\\ =&
 -\flat_{L} \left(\mathcal{Q}\left( \Gamma_{L} \right) \right) \in \ann{\Delta^{l}}
\end{align*}

\noindent{Therefore, by uniqueness, $X_{\vert \Delta \times\RR} = \Gamma_{L , \Delta}$ is a solution of Eq. (\ref{6}).}\\\\
Let us now compute an explicit expression of the solution $\Gamma_{L,\Delta}$. Let $Y$ be a vector field on $TQ \times \RR$. Then, choosing a local basis $\{ \beta_{i}\}$ of $T \left( \Delta \times \RR \right)$ we may write the restriction of $Y$ to $\Delta \times \RR$ as follows
\[ Y_{\vert \Delta \times \RR} = Y^{i}\beta_{i} + \lambda^{a}Z_{a}.\]
Then, applying $\dd \overline{\Phi}^{b}$ we have that
\[ \dd \overline{\Phi}^{b} \left( Y \right) = \lambda^{a} C_{ba},\]
and we can compute the coefficients $\lambda^{a}$ as follows
\begin{equation}\label{19}
\lambda^{a} = C^{ba}\dd \overline{\Phi}^{b} \left( Y\right)
\end{equation}
Hence, for all vector field $Y$ on $TQ \times \RR$ restricted to $\Delta \times \RR$
\begin{itemize}
\item $\mathcal{Q} \left( Y_{\vert \Delta \times \RR} \right) = C^{ba}\dd \overline{\Phi}^{b} \left( Y \right)Z_{a}.$

\item $\mathcal{P} \left( Y_{\vert \Delta \times \RR} \right) = Y_{\vert \Delta \times \RR} - C^{ba}\dd \overline{\Phi}^{b} \left( Y \right)Z_{a}.$
\end{itemize}

Therefore, we have obtained the explicit expression of the solution $\Gamma_{L, \Delta}$,
\begin{equation}\label{20}
\Gamma_{L,\Delta} = \left(\Gamma_{L}\right)_{\vert \Delta \times \RR} - C^{ba}\dd \overline{\Phi}^{b} \left( \Gamma_{L} \right)Z_{a}
\end{equation}
\begin{remark}
\rm
From the regularity of the matrices $C$ , we deduce that the projections $\mathcal{P}$ and $\mathcal{Q}$ may be extended to open neighborhoods of $\Delta \times \RR$. Consequently, $\mathcal{P} \left( \Gamma_{L} \right)$ may also be extended to an open neighborhood of $\Delta \times \RR$. However, this extension will not be unique.\\
\end{remark}

From the projectors $\mathcal{P}$ and $\mathcal{Q}$ defined in Eq. (\ref{projectors2}) we may construct a new pair of projectors $\overline{\mathcal{P}}$ and $\overline{\mathcal{Q}}$ acting on the covectors. These projectors will transform \cref{6} into an exact equation (see theorem \ref{21}).

Notice that, as a consequence of the regularity we have that
$$ T^{*}_{\Delta \times \RR}\left( TQ \times \RR \right) = \overline{\mathcal{S}} \oplus \ann{\Delta^{l}},$$
where $\overline{\mathcal{S}} = \flat_{L}\left( T \left(\Delta \times \RR \right)\right)$ and $T^{*}_{\Delta \times \RR} \left( TQ \times \RR \right)$ are the $1-$forms on $TQ \times \RR$ at points of $\Delta \times \RR$. Notice that, by construction, $\flat_{L}\left( \mathcal{S} \right) = \ann{\Delta^{l}}$. Then, for all $a_{\left( v_{q},z\right)} \in T^{*}_{\left( v_{q},z\right)} \left( TQ \times \RR \right)$ with $\left( v_{q},z\right) \in \Delta \times \RR$, the associated projections $\overline{\mathcal{Q}}: T^{*}_{\Delta \times \RR}\left( TQ \times \RR \right) \rightarrow \ann{\Delta^{l}}$ and $\overline{\mathcal{P}}: T^{*}_{\Delta \times \RR}\left( TQ \times \RR \right) \rightarrow \overline{\mathcal{S}}$ are given by
\begin{itemize}
\item $\overline{\mathcal{Q}} \left( a_{\left( v_{q},z\right)}\right) = \flat_{L}\left(\mathcal{Q}\left(\sharp_{L} \left( a_{\left( v_{q},z\right)}\right)\right)\right)$.

\item $\overline{\mathcal{P}} \left( a_{\left( v_{q},z\right)}\right) = \flat_{L}\left(\mathcal{P}\left(\sharp_{L} \left( a_{\left( v_{q},z\right)}\right)\right)\right)$.
\end{itemize}
So, the constraint Herglotz equations may be rewritten as follows
\begin{theorem}\label{21}
Assume that $L$ is regular. Let $X$ be a vector field on $TQ \times \RR$ satisfying the equations
\begin{equation}
    \begin{cases}\label{22}
        \flat_{L}\left( X \right) = \overline{\mathcal{P}} \left( \dd E_L + \left(E_L + \Reeb_L\left(E_L\right)\right)\eta_L\right) \\
        X_{\vert \Delta \times \RR} \in  \frak X \left(\Delta \times \RR \right).
    \end{cases}
\end{equation}
Then, the integral curves of $X$ are solutions of the constraint Herglotz equations (\ref{eq:nonholonomic_herglotz_eqs_coords}), that is, $X = \Gamma_{L, \Delta}$ is the solution of Eq. (\ref{6}).
\end{theorem}

Let us recall that the contact Hamiltonian vector fields model the dynamics of dissipative
systems and, contrary to the case of symplectic Hamiltonian systems, the evolution does not preserve the energy, the contact form and the volume (\cref{prop:energy_disipation}), i.e.,
\begin{align*}
    \lieD{\Gamma_L} E_L &= -\Reeb_L (E_L) E_L, \\
    \mathcal{L}_{\Gamma_{L}} \eta_{L} &= -\mathcal{R}_{L} \left( E_{L}\right)\eta_{L}.
\end{align*}
This result may be naturally generalized to the case of non-holonomic constraint by using these projectors.
\begin{proposition}\label{39}
Assume that $L$ is regular. The vector field $\Gamma_{L, \Delta}$ solving the constraint Herglotz equations satisfies that

\begin{subequations}
    \begin{align}
        \mathcal{L}_{\Gamma_{L,\Delta}} \eta_{L} &= 
          - \mathcal{R}_{L} \left( E_{L}\right)\eta_{L}- \lieD{\mathcal{Q}(\Gamma_{L})} \eta_{L},\\
        \mathcal{L}_{\Gamma_{L,\Delta}} \tilde{\eta}_{L} &= 
          - \frac{\lieD{\mathcal{Q}(\Gamma_{L})} \eta_{L}}{H} \\
        \lieD{\Gamma_{L,\Delta}} \Omega_{L} &= -(n+1)\Reeb_L(E_L) \Omega_L  - \eta_L \wedge {\dd \eta_L}^{(n-1)} \wedge \dd \mathcal{L}_{\mathcal{Q}(\Gamma_{L})} \eta_{L}\\
        \lieD{\Gamma_{L,\Delta}} \tilde{\Omega}_{L} &= \tilde{\eta}_L \wedge {\dd \tilde{\eta}_L}^{(n-1)} \wedge \dd \mathcal{L}_{\mathcal{Q}(\Gamma_{L})} \tilde{\eta}_{L}
    \end{align}
\end{subequations}

where $\tilde{\eta}_L = \eta_L/H$, assuming that $H$ does not vanish,  $\Omega_L = \eta_L \wedge {(\dd \eta_L)}^{n}$ is the contact volume element and $\tilde{\Omega}_L = \eta_L \wedge {(\dd \eta_L)}^{n}$.

Furthermore, we have that $\lieD{\mathcal{Q}\left(\Gamma_{L}\right) \eta_{L}} \in  \ann{\Delta^{l}} $.
\begin{proof}
The first fact follows from $\mathcal{L}_{\Gamma_{L}} \eta_{L} = -\mathcal{R}_{L} \left( E_{L}\right)\eta_{L}$.

The second claim follows from the product rule.

For the third claim, we perform the following computation.
\begin{small}
\begin{equation*}
    \begin{aligned}
        \lieD{\Gamma_{L,\Delta}}( \eta_L \wedge (\dd \eta_L)^n ) &= 
        - \Reeb(E_L) \eta_L \wedge (\dd \eta_L)^n  + 
        n \eta_L \wedge (\dd \eta_L)^{n-1} \wedge \dd \lieD{\Gamma_{L,\Delta}}(\eta_L ) \\ &= 
        -(n+1) \Reeb(E_L) \eta_L \wedge (\dd \eta_L)^n -
        \eta_L \wedge {\dd \eta_L}^{(n-1)} \wedge \dd \mathcal{L}_{\mathcal{Q}(\Gamma_{L})} \eta_{L},
    \end{aligned}
\end{equation*}
\end{small}
A similar computation proves the forth claim.

The last assertion is proved by writing $\mathcal{Q}\left(\Gamma_{L}\right)$ in coordinates depending on the vector fields $Z_{a}$.
\end{proof}

\end{proposition}

\section{Non-holonomic bracket}\label{sec:nonholonomic_brackets}
Consider a regular contact Lagrangian system with Lagrangian $L:TQ \times \RR \to \RR$ and constraints $\Delta$ satisfying the conditions in \cref{sec:herglotz_constrained}. A bracket can be constructed by means of the decomposition~\cref{41}.

Let us first consider the adjoint operators $\mathcal{P}^*$ and $\mathcal{Q}^*$ of the projections $\mathcal{P}$ and $\mathcal{Q}$, respectively. Obviously, the maps $\mathcal{P}^* :T^{*}_{\Delta \times \RR} \left( TQ \times \RR \right) \rightarrow \ann{\mathcal{S}}$ and $\mathcal{Q}^* :T^{*}_{\Delta \times \RR} \left( TQ \times \RR \right) \rightarrow \ann{T} \left( \Delta \times \RR \right)$ produce a decomposition of $T^{*}_{\Delta \times \RR} \left( TQ \times \RR \right)$
\begin{equation}\label{43}
T^{*}_{\Delta \times \RR} \left( TQ \times \RR \right) = \ann{\mathcal{S}}\oplus \ann{T} \left( \Delta \times \RR \right)
\end{equation}

Notice that, using the proof of \cref{18}, we know that a solution $\Gamma_{L, \Delta}$ of the \cref{6} may be written (locally) as follows
\[ \Gamma_{L, \Delta} = \dot{q}^{i} \pdv{}{q^{i}} + \Gamma_{L , \Delta} \left( \dot{q}^{i}\right) \pdv{}{\dot{q}^{i}} + L \pdv{}{z},\]
where $\Gamma_{L,\Delta}\left( \dot{q}^{i}\right)$ satisfies that

\begin{equation*}
    \Gamma_{L,\Delta}\left( \dot{q}^{k}\right) \dfrac{\partial^{2}L}{\partial\dot{q}^k \partial \dot{q}^{i}} = - \dot{q}^{k}\dfrac{\partial^{2}L}{\partial  q^k \partial \dot{q}^{i}} - L \dfrac{\partial^{2}L}{\partial z \partial \dot{q}^{i}} + \dfrac{\partial  L }{\partial q^{i}} +\dfrac{\partial  L }{\partial \dot{q}^{i}}\dfrac{\partial  L }{\partial z} +\lambda_a\Phi^a_i
\end{equation*}

On the other hand, recall that the codistribution $\ann{\Delta^{l}}$ is generated by the 1-forms on $TQ \times \mathbb{R}$ given by $\tilde{\Phi}^{a}=\Phi^{a}_{i}{\dd q}^{i}$. Thus, for each $\left( v_{q} , z \right) \in T_{q}Q \times \mathbb{R}$ we have that
\begin{eqnarray*}
\left[ \tilde{\Phi}^{a}\left( \Gamma_{L, \Delta} \right) \right]  \left( v_{q} , z \right)  & = &  \left[ \dot{q}^{i} \Phi^{a}_{i}\right]  \left( v_{q} , z \right)\\
& = &  v^{i}_{q}\Phi^{a}_{i}  \left( q\right)\\
& = & \Phi^{a} \left( v_{q}\right)
\end{eqnarray*}
with $v^{i}_{q} = \dot{q}^{i} \left( v_{q}\right) $ for all $i$. Hence, by construction (see Eq. (\ref{5})), we have that for each $\left( v_{q} , z \right) \in \Delta_{q} \times \mathbb{R}$
$$\{\tilde{\Phi}^{a}\left( \Gamma_{L, \Delta} \right) \}  \left( v_{q} , z \right)  = \Phi^{a} \left( v_{q}\right) = 0,$$
i.e.,
\begin{equation}\label{24}
\Gamma_{L , \Delta} \in \Delta^{l}.
\end{equation}
Then, we have proved that any solution $\Gamma_{L, \Delta}$ to Eq. (\ref{6}) is tangent to the intersection  of $T \left( \Delta \times \mathbb{R} \right)$ with $\Delta^{l}$.\\
\begin{theorem}\label{25}
A vector field $X$ on $TQ \times \mathbb{R}$ satisfies
$$\flat_L\left(X\right) - \dd E_L + \left(E_L + \Reeb_L\left(E_L\right)\right)\eta_L \in \ann{\Delta^{l}}$$
if, and only if,
\begin{equation}
    \begin{cases}\label{26}
    \mathcal{L}_{X}\eta_{L} + \Reeb_L\left(E_{L}\right) \eta_{L} \in \ann{\Delta^{l}} \\
        \eta_{L} \left(  X\right) = - E_{L}.
    \end{cases},
\end{equation}
\begin{proof}
Assume that $X $ fulfills
$$\flat_L\left(X\right) - \dd E_L + \left(E_L + \Reeb_L\left(E_L\right)\right)\eta_L \in \ann{\Delta^{l}}.$$
Then, applying $\sharp_{L}$ on both sides, we have that
\begin{equation}\label{40}
X - \sharp_{L}\left(\dd E_L\right) + \left(E_L + \Reeb_L\left(E_L\right)\right)\Reeb_L \in \mathcal{S}.
\end{equation}
Now, let us apply $\eta_{L}$ to \cref{40} and, using that $\eta_L (\mathcal{S}) = 0$ (\cref{rem:S}), then we conclude that
\begin{equation}\label{30}
\eta_{L}\left( X \right)  - \eta_{L} \left( \sharp_{L}\left(\dd E_L\right)\right) + \left(E_L + \Reeb_L\left(E_L\right)\right) = 0
\end{equation}
Observe that, by definition
$$\iota_{\sharp_{L}\left(\dd E_L\right)} \dd \eta_{L} +\eta_{L} \left( \sharp_{L}\left(\dd E_L\right) \right) \eta_{L} = \dd E_{L}.$$
So, 
\begin{equation}\label{27}
\dd E_{L} \left( \Reeb_L \right) = \Reeb_L\left( E_{L}\right) = \eta_{L} \left( \sharp_{L}\left(\dd E_L\right) \right).
\end{equation}
Therefore, Eq. (\ref{30}) turns into the following equation
\begin{equation}\label{31}
\eta_{L}\left( X \right)  =- E_{L}
\end{equation}
On the other hand,
\begin{eqnarray*}
\mathcal{L}_{X} \eta_{L} & = & \iota_{X} \dd \eta_{L} + \dd \left(\eta_{L}\left( X \right)\right)\\
& = &  \flat_L \left( X \right) - \eta_{L}\left( X \right) \eta_{L} - \dd E_{L}\\
& = &  \flat_L \left( X \right) + E_{L} \eta_{L} - \dd E_{L}\\
\end{eqnarray*} 
Then,
\begin{eqnarray*}
\mathcal{L}_{X}\eta_{L} + \Reeb_L\left(E_{L}\right) \eta_{L} & = & \flat_L \left( X \right) + E_{L} \eta_{L} - \dd E_{L} + \Reeb_L\left(E_{L}\right) \eta_{L}\\
&=& \flat_L \left( X \right) - \dd E_{L} + \left( E_{L} + \Reeb_L\left(E_{L}\right)\right) \eta_{L} \in \ann{\Delta^{l}}
\end{eqnarray*}
The other implication is proved by using this equation next to Eq. (\ref{31}).
\end{proof}
\end{theorem}
Then, as a corollary, we have another geometric equation equivalent to Eq. (\ref{6}).

\begin{corollary}\label{35}
A vector field $X$ on $TQ \times \mathbb{R}$ satisfies Eq. (\ref{6}) if, and only if,
\begin{equation}
    \begin{cases}\label{32}
    \mathcal{L}_{X}\eta_{L} + f \eta_{L} \in \ann{\Delta^{l}} \\
        \eta_{L} \left(  X\right) = - E_{L}\\
        X_{\vert \Delta \times \RR} \in  \frak X \left(\Delta \times \RR \right).
    \end{cases},
\end{equation}
for some function $f$.
\begin{proof}
Notice that $\Reeb_{L} \in \Delta^{l}$. Then,
$$0 = \left(\mathcal{L}_{X}\eta_{L} + f \eta_{L} \right) \left( \Reeb_{L} \right)= \mathcal{L}_{X}\eta_{L} \left( \Reeb_{L}\right) + f .$$
Therefore,
\begin{eqnarray*}
f &=& - \mathcal{L}_{X}\eta_{L} \left( \Reeb_{L} \right) \\
&=&- \left(\dd \left( \eta_{L} \left(  X\right) \right) \right) \left( \Reeb_{L} \right) - \left(\iota_{X} \dd \eta_{L} \right) \left( \Reeb_{L} \right)\\
 & = & \left(\dd E_{L} \right) \left( \Reeb_{L} \right)\\
 &=& \Reeb_{L} \left( E_{L }\right)
\end{eqnarray*}

\end{proof}
\end{corollary}
Compare this to \cref{eq:ham_vf_conf_contactomorphism}, for the case without constraints. Notice the similarity with proposition \ref{39}.\\

Let us now prove a technical but necessary lemma.
\begin{lemma}\label{37}
The following identity holds
$$\mathcal{Q}^{*}\left( \dd E_{L}  \right)= 0,$$
i.e., $\dd E_{L} \in \ann{\mathcal{S}}$.

\begin{proof}
Taking into account $\mathcal{S} \subseteq \Delta^{l}$ (\ref{33}), we have that 
\begin{equation}\label{34}
\ann{\Delta^{l}}\subseteq \ann{\mathcal{S}}
\end{equation}
Let $\Gamma_{L,\Delta}$ be a solution of Eq. (\ref{6}). Then, using corollary~\ref{35}, we get
$$\mathcal{L}_{\Gamma_{L, \Delta}s}\eta_{L} + \Reeb_L\left(E_{L}\right) \eta_{L} \in \ann{\mathcal{S}}$$
Hence, projecting by $\mathcal{Q}^{*}$,
$$\mathcal{Q}^{*}\left( \mathcal{L}_{\Gamma_{L, \Delta}s}\eta_{L} \right) = 0,$$
i.e.,
\begin{equation}\label{36}
\mathcal{Q}^{*}\left( \iota_{\Gamma_{L,\Delta} }\dd \eta_{L}  \right) =  -\mathcal{Q}^{*}\left( \dd\left( \eta_{L}\left( \Gamma_{L,\Delta}\right)\right) \right) =  \mathcal{Q}^{*}\left( \dd E_{L}  \right) 
\end{equation}
However, for any other vector field $Y$ on $TQ \times \mathbb{R}$ restricted to $\Delta \times \mathbb{R}$,
\begin{equation*}
\left[ \mathcal{Q}^{*}\left( \iota_{\Gamma_{L,\Delta} }\dd \eta_{L}  \right) \right] \left( Y \right) = \dd \eta_{L}\left( \Gamma_{L, \Delta} , \mathcal{Q} \left( Y \right) \right) =\left[ \dd \eta_{L } + \eta_{L}\otimes\eta_{L}\right] \left( \Gamma_{L, \Delta} , \mathcal{Q} \left( Y \right) \right) = 0.
\end{equation*}
The second and third equalities are consequence of \cref{rem:S} and \cref{24}. Then
$$\mathcal{Q}^{*}\left( \dd E_{L}  \right)= 0.$$
\end{proof}
\end{lemma}
Notice that, as a immediate consequence, we have that
\begin{equation}\label{44}
\dd E_L =  \mathcal{P}^* \left( \dd E_L \right)
\end{equation}

We may now define, along $\Delta \times \RR$, the following vector and bivector fields:
\begin{align}\label{almostjac24}
    \Reeb_{L,\Delta} &=  \mathcal{P} \left({\Reeb_L}_{\vert \Delta \times \RR}\right),\\
    \Lambda_{L,\Delta} &= \mathcal{P}_* {\Lambda_L}_{\vert \Delta \times \RR},
\end{align}
where $\Lambda_{L}$ is the Jacobi structure associated to the contact form $\eta_{L}$ (see \cref{eq:contact_jacobi}). That is, for $\left( v_{q} , z \right) \in \Delta\times \RR \subseteq TQ\times \RR$ and $\alpha,\beta\in T_{\left( v_{q} , z \right)}^* \left(TQ\times \RR\right)$,   
$$\Lambda_{L,\Delta}\left(\alpha,\beta\right) =  \Lambda_{L}\left(\mathcal{P}^*\left(\alpha\right),\mathcal{P}^*\left(\beta\right)\right).$$

\noindent{This structure provides the following morphism of vector bundles}
\begin{equation}
  \begin{aligned}
    \sharp_{\Lambda_{L,\Delta}}:  T^{*}_{\Delta \times \RR}\left( TQ \times \RR \right) &\to T_{\Delta \times \RR}\left( TQ \times \RR \right),\\
    \alpha &\mapsto \Lambda_{L,\Delta}(\alpha, \cdot).
  \end{aligned}
\end{equation}

Hence, we may prove the following result:
\begin{lemma}\label{45}
  \begin{equation}
    \mathcal{P}(\sharp_{\Lambda_{L}}(\dd {E_L})) = \sharp_{\Lambda_{L,\Delta}}(\dd {E_L}).
  \end{equation}
\end{lemma}
\begin{proof}
Let us consider an arbitrary $\alpha\in T_{\left( v_{q} , z \right)}^* (TQ\times \RR)$  with $\left( v_{q} , z \right) \in \Delta \times \RR$,
  \begin{align*}
    \alpha \left( \mathcal{P}\left( \sharp_{\Lambda_{L}} \left( \dd {E_L}_{\vert \left( v_{q} , z \right)} \right) \right) \right)  &=
    \left[  \mathcal{P}^*  \alpha \right] \left(\sharp_{\Lambda_{L}}\left(\dd {E_L}_{\vert \left( v_{q} , z \right)} \right) \right)  \\ &=
    \Lambda_{L}(\dd {E_L}_{\vert \left( v_{q} , z \right)}, \mathcal{P}^* (\alpha)),
  \end{align*}
Then, using \cref{44}, we get
\begin{align*}
 \alpha \left( \mathcal{P}\left( \sharp_{\Lambda_{L}} \left( \dd {E_L}_{\vert \left( v_{q} , z \right)} \right) \right) \right)  &=
  \Lambda_{L}(\mathcal{P}^* (\dd E_L), \mathcal{P}^* (\alpha)) \\ 
  &=  \sharp_{\Lambda_{L,\Delta}}(\dd E_L)(\alpha),
 \end{align*}
 and the result follows.
\end{proof}

\begin{theorem}\label{46}
We have
  \begin{equation}
    \Gamma_{L,\Delta} = \sharp_{\Lambda_{L,\Delta}}(\dd {E_L}) - {E_L} {\Reeb_{L,\Delta} }
  \end{equation}
\end{theorem}
\begin{proof}
  Along $\Delta \times \RR$, we have
  \begin{align*}
    \Gamma_{\Delta,L} &= \mathcal{P}(\Gamma_{L}) \\ &= 
    \mathcal{P}(\sharp_{\Lambda_L}(\dd E_L) -  E_L \Reeb_L) \\ &=
    {\sharp_{\Lambda_{L,\Delta}}}(\dd E_L) - E_L \Reeb_{L,\Delta}, 
  \end{align*}
  where we have used~\cref{45}.
\end{proof}

Furthermore, we can define the following bracket from functions on $TQ \times \RR$ to functions on $\Delta \times \RR$, which will be called the \emph{nonholonomic bracket}:
\begin{equation}
  \NHBr{f,g} = \Lambda_{L,\Delta}(\dd f, \dd g) - f \Reeb_{L,\Delta}(g) + g \Reeb_{L,\Delta}(f).
\end{equation}

\begin{theorem}\label{esteteorema23}
  The nonholonomic bracket has the following properties:
  \begin{enumerate}
    \item Any function $g$ on $TQ\times \RR$ that vanishes on $\Delta \times \RR$ is a Casimir, i.e.,
    \[\NHBr{g,f} = 0, \ \forall f \in \mathcal{C}^{\infty} \left( TQ \times \RR \right)\]
    \item The bracket provides the evolution of the observables, that is, 
    \begin{equation}
      \Gamma_{L,\Delta}(g) = \NHBr{E_L,g} - g \Reeb_{L,\Delta} (E_L).
    \end{equation}
  \end{enumerate}
\end{theorem}
\begin{proof}
For the first assertion, let $g$ a function which vanish on $\Delta \times \RR$ and let $f$ be any function on $TQ \times \RR$. Notice that this implies that $\dd g \in \ann{(T(\Delta \times \RR))}$, hence $\mathcal{P}^* (\dd g) = 0$. Then, along $\Delta \times \RR$, we have that
  \begin{align*}
    \NHBr{g,f} &= 
    \Lambda_{L,\Delta}(\dd g,\dd f) - g \Reeb_{L,\Delta} (f) + f \Reeb_{L,\Delta} (g) \\ &= 
    \Lambda_{L}(\mathcal{P}^*(\dd g),\mathcal{P}^*(\dd f)) - g \Reeb_{L,\Delta} (f) + f \Reeb_{L,\Delta} (g) = 0.
  \end{align*}

For the second part, notice that

\begin{equation*}
        \NHBr{E_L,g} - g \Reeb_{L,\Delta} (E_L) =
        \Lambda_{L,\Delta}(\dd E_L, \dd g) -  E_L \Reeb_{L,\Delta}(g) =
        \Gamma_{L,\Delta}(g),    
\end{equation*}
where we have used~\cref{46}.
\end{proof}
Notice that, in particular, all the constraint functions $\Phi^{a}$ are Casimir.\\
It is also remarkable that, using the statement \textit{1.} in \cref{esteteorema23}, the nonholonomic bracket may be restricted to functions on $\Delta \times \RR$. Thus, from now on, we will refer to the nonholonomic bracket as the restriction of $\NHBr{\cdot , \cdot }$ to functions on $\Delta \times \RR$.

\section{Hamiltonian vector fields and integrability conditions}\label{quizalapenultimaseccion}
Until now, we have defined a structure given by a vector field $ \Reeb_{L,\Delta}$ and a bivector field $\Lambda_{L,\Delta} $ (see \cref{almostjac24}) which induce the nonholonomic bracket \eqref{nonholbracket243}
\begin{equation}
  \NHBr{f,g} = \Lambda_{L,\Delta}(\dd f, \dd g) - f \Reeb_{L,\Delta}(g) + g \Reeb_{L,\Delta}(f).
\end{equation}
This structure is quite similar to a Jacobi structure (see \cref{sec:contact_jacobi}). In fact, we may prove the following result.\\
\begin{proposition}
The nonholonomic bracket endows the space of differentiable functions on $\Delta \times \RR$ with an almost Lie algebra structure \cite{daSilva1999} which satisfies the generalized Leibniz rule
        \begin{equation}\label{eq:mod_leibniz_rulenonhol}
           \NHBr{f,gh} = g\NHBr{f,h} + h\NHBr{f,g} -  g h \Reeb_{L, \Delta}(h),
        \end{equation}
\begin{proof}
The nonholonomic bracket obviously satisfies that it is $\RR-$linear and skew-symmetric. The Leibniz rule follows from a straightforward computation.
\end{proof}

\end{proposition}
So, as an obvious corollary we have that
\begin{corollary}
The vector field $ \Reeb_{L,\Delta}$ and the bivector field $\Lambda_{L,\Delta} $ induce a Jacobi stucture on $\Delta \times \RR$ if, and only if, the nonholonomic bracket satisfies the Jacobi identity.
\end{corollary}
This result motivates the following definition.
\begin{definition}
Let $M$ be a manifold with a vector field $E$ and a bivector field $\Lambda$. The triple $(M,\Lambda,E)$ is said to be an \textit{almost Jacobi structure} if the pair $(\Cont^\infty(M),\jacBr{\cdot,\cdot})$ is an almost Lie algebra satisfying the generalized Leibniz rule (\ref{eq:mod_leibniz_rulenonhol}) where the bracket is given by 
\begin{equation}\label{nonholbracket243}
  \jacBr{f,g} = \Lambda(\dd f, \dd g) + f E(g) - g E(f)
\end{equation}
\end{definition}
With this, the triple $\left( \Delta \times \RR,  \Lambda_{L,\Delta}, -\Reeb_{L,\Delta}  \right)$ is an almost Jacobi structure. Of course, the study of the intrisic properties of almost Jacobi structures on general manifolds has a great interest from the mathematical point of view. However, this could distract the reader from the main goal of this paper. So, here we will only focus on the necessary properties for our develoment.\\

Let $H$ be a Hamiltonian function on the contact manifold $\left(TQ \times \RR,\eta_{L}\right)$. Then, we define the \emph{Constrained Hamiltonian vector field} $X_H^{\Delta}$ by the equation
\begin{equation}\label{eq:nonholhamiltonian_vf_contact}
    X_H^{\Delta} = \sharp_{\Lambda_{L,\Delta}} \left(\dd H \right) - H \Reeb_{L, \Delta} .
\end{equation}
Then, by using \cref{46} we have that the solution $\Gamma_{L, \Delta}$ of \cref{6} is a particular case of constrained Hamiltonian vector field. In fact,
$$\Gamma_{L, \Delta} = X_{E_{L}}^{\Delta}.$$

As in the case without constraints, we have many equivalent ways of defining these vector fields.

\begin{proposition}\label{anotherpropomore23}
Let $H$ be a Hamiltonian function on $TQ \times \RR$. The following statements are equivalent:
\begin{itemize}
\item[(i)] $X_{H}^{\Delta}$ is the Constraind Hamiltonian vector field of $H$.

\item[(ii)] It satisfies the following equation,
\begin{equation}\label{eq:nonholhamiltonian_vf_contactsecond}
    X_H^{\Delta} = \mathcal{P}\left( \sharp_{L}\left(\mathcal{P}^{*}\dd H \right)\right) - \left(\Reeb_{L,\Delta} \left(H\right) + H\right) \Reeb_{L,\Delta}.
\end{equation}
\item[(iii)] The following equation holds,
\begin{equation}\label{eq:nonholhamiltonian_vf_contactthird}
    X_{H}^{\Delta} = \mathcal{P}\left( X_{H}\right) - \mathcal{P}\left(   \sharp_{\Lambda_{L}}\left( \mathcal{Q}^{*}\dd H \right)  \right).
\end{equation}
\end{itemize}
\begin{proof}
Let $g$ a smooth function of $TQ \times \RR$. Then,
\begin{eqnarray*}
X_{H}^{\Delta} \left( g \right) &=&  \{\sharp_{\Lambda_{L,\Delta}} \left(\dd H \right) \}\left( g \right)- H \Reeb_{L, \Delta} \left( g \right)\\
&=&  \{\mathcal{P}\left( \sharp_{\Lambda_{L}} \left(\mathcal{P}^{*}\dd H \right)\right) \}\left( g \right) - H \Reeb_{L, \Delta}  \left( g \right)\\
&=&  \{\mathcal{P}\left( \sharp_{L} \left(\mathcal{P}^{*} \dd H \right) - \mathcal{P}^{*} \dd H \left( \Reeb_{L}\right) \Reeb_{L} \right)\}\left( g \right) - H \Reeb_{L, \Delta}  \left( g \right)\\
&=&  \{\mathcal{P}\left( \sharp_{L} \left(\mathcal{P}^{*} \dd H \right) - \Reeb_{L, \Delta}\left( H\right) \Reeb_{L, \Delta} \right) \}\left( g \right) - H \Reeb_{L, \Delta}  \left( g \right)\\
&=& \mathcal{P}\left( \sharp_{L}\left(\dd H \right)\right) \left( g \right)- \left(\Reeb_{L,\Delta} \left(H\right) + H\right) \Reeb_{L,\Delta} \left( g \right)
\end{eqnarray*}
This proves that $(i)$ is equivalent to $(ii)$. The equivalence between $(i)$ and $(iii)$ follows using the natural decomposition of $\sharp_{\Lambda_{L}}\left( \dd H \right)$ into $\sharp_{\Lambda_{L}}\left( \mathcal{P}^{*}\dd H \right)$ and $\sharp_{\Lambda_{L}}\left( \mathcal{Q}^{*}\dd H \right)$.
\end{proof}
\end{proposition}
Notice that the constrained Hamiltonian vector field $X_{H}^{\Delta}$ is just a vector field along the submanifold $\Delta \times \RR$.
\begin{corollary}
Let $H$ be a Hamolitonian function on $TQ \times \RR$. Then, it satisfies that
\begin{equation}\label{etaproyect34}
\eta_{L}(X^{\Delta}_H) = -H.
\end{equation}
\begin{proof}

By using \cref{prop213123}, \cref{rem:S} and \cref{anotherpropomore23}, we have that
\begin{eqnarray*}
\eta_{L} \left( X^{\Delta}_H  \right) & = &  \eta_{L} \left(  \mathcal{P}\left( X_{H}\right) - \mathcal{P}\left(   \sharp_{\Lambda_{L}}\left( \mathcal{Q}^{*}\dd H \right)  \right) \right)\\
& = &  \eta_{L} \left(   X_{H} -   \sharp_{\Lambda_{L}}\left( \mathcal{Q}^{*}\dd H \right)  \right)\\
& = &  -H -\eta_{L} \left( \sharp_{\Lambda_{L}}\left( \mathcal{Q}^{*}\dd H \right)  \right)
\end{eqnarray*}
On the other hand,
\begin{eqnarray*}
\eta_{L} \left( \sharp_{\Lambda_{L}}\left( \mathcal{Q}^{*}\dd H \right)  \right)  &=&  \eta_{L} \left( \sharp_{L}\left( \mathcal{Q}^{*}\dd H \right)  -  \{ \mathcal{Q} \left( \Reeb_{L}\right) \} \left( H \right)\Reeb_{L}\right)\\ 
&=&  \eta_{L} \left( \sharp_{L}\left( \mathcal{Q}^{*}\dd H \right) \right) -  \{ \mathcal{Q} \left( \Reeb_{L}\right) \} \left( H \right) = 0.
\end{eqnarray*}
\end{proof}

\end{corollary}
As a consequence of this corollary we have that the correspondence $H \mapsto X_{H}^{\Delta}$ is, in fact, an isomorphism of vector spaces. By means of this isomoprhism, we may prove the following result.
\begin{proposition}\label{yotramasss}
The nonholonomic bracket $\NHBr{\cdot , \cdot}$ satisfies the Jacobi identity if, and only if, 
$$ \left[ X^{\Delta}_{F}, X^{\Delta}_{G} \right] = X^{\Delta}_{ \NHBr{F,G}},$$
i.e., the correspondence $H \mapsto X^{\Delta}_{H}$ is an isomorphism of Lie algebras.
\begin{proof}
Let us consider four arbitrary functions $F,G,H,f \in \Cont^\infty(M)$. Then,

\begin{scriptsize}

\begin{eqnarray*}
\NHBr{fF, \NHBr{G,H}} &=& f \NHBr{F, \NHBr{G,H}} + F \NHBr{f, \NHBr{G,H}} + fF\Reeb_{L, \Delta} \left( \NHBr{G,H} \right)
\end{eqnarray*}

\begin{eqnarray*}
\NHBr{H, \NHBr{fF,G}} &=&  \NHBr{H, f\NHBr{F,G}} +  \NHBr{H, F\NHBr{f,G}} + \NHBr{H , fF\Reeb_{L, \Delta} \left( G \right) }\\
&=&   f\NHBr{H, \NHBr{F,G}} +   \NHBr{F,G}\NHBr{H,f}  -  f \NHBr{F,G} \Reeb_{L , \Delta} \left( H \right)\\
&   &  +F\NHBr{H, \NHBr{f,G}} +   \NHBr{f,G}\NHBr{H,F}  -  F \NHBr{f,G} \Reeb_{L , \Delta} \left( H \right)\\
&    &   +fF \NHBr{H, \Reeb_{L, \Delta} \left(G \right)}  +  \Reeb_{L,\Delta} \left(G \right) \left[   f \NHBr{H,F} + F \NHBr{H,f}  - fF \Reeb_{L,\Delta} \left( H \right)\right]
\end{eqnarray*}

\begin{eqnarray*}
\NHBr{G, \NHBr{H,fF}} &=&  -\NHBr{G, f\NHBr{F,H}} -  \NHBr{G, F\NHBr{f,H}} - \NHBr{G , fF\Reeb_{L, \Delta} \left( H \right) }\\
&=&   -f\NHBr{G, \NHBr{F,H}} -   \NHBr{F,H}\NHBr{G,f}  +  f \NHBr{F,H} \Reeb_{L , \Delta} \left( G \right)\\
&   &  -F\NHBr{G, \NHBr{f,H}} -  \NHBr{f,H}\NHBr{G,F}  +  F \NHBr{f,H} \Reeb_{L , \Delta} \left( G \right)\\
&    &  - fF \NHBr{G, \Reeb_{L, \Delta} \left(H \right)}  -  \Reeb_{L,\Delta} \left(H \right) \left[   f \NHBr{G,F} - F \NHBr{G,f}  + fF \Reeb_{L,\Delta} \left( G \right)\right]
\end{eqnarray*}

\end{scriptsize}
So, adding these equalities we have that,
\begin{scriptsize}
\begin{equation}\label{eq:otramas}
\begin{aligned}
& \NHBr{fF, \NHBr{G,H}}  + \NHBr{H, \NHBr{fF,G}} +  \NHBr{G, \NHBr{H,fF}} = \\
&  - f F \Reeb_{L , \Delta}  \left(  \NHBr{G,H} \right) - fF \left[\NHBr{H, \Reeb_{L, \Delta} \left(G \right)}  - \NHBr{G, \Reeb_{L, \Delta} \left(H \right)} \right] \\
&  + f \left[ \NHBr{F, \NHBr{G,H}} + \NHBr{H, \NHBr{F,G}} +\NHBr{G, \NHBr{H , F}}  \right]\\
&  + F \left[ \NHBr{f, \NHBr{G,H}}  + \NHBr{H, \NHBr{f,G}} + \NHBr{G, \NHBr{H,f}}\right]
\end{aligned}
\end{equation}

\end{scriptsize}

On the other hand,
\begin{scriptsize}
\begin{eqnarray*}
X_{F}^{\Delta}\left( X_{G}^{\Delta} \left( H \right) \right) &=&  \NHBr{F, \NHBr{G,H}} - H\NHBr{F, \Reeb_{L, \Delta}\left(G \right)} - \Reeb_{L, \Delta}\left(G \right) \NHBr{F,H} \\
&  & + H \Reeb_{L, \Delta}\left(G \right)\Reeb_{L, \Delta}\left(F \right) - \left[ \NHBr{G,H}  - H \Reeb_{L, \Delta}\left(G \right)  \right]\Reeb_{L, \Delta}\left(F \right)
\end{eqnarray*}
\end{scriptsize}
Hence,
\begin{scriptsize}
\begin{eqnarray*}
X_{F}^{\Delta}\left( X_{G}^{\Delta} \left( H \right) \right) - X_{G}^{\Delta}\left( X_{H}^{\Delta} \left( H \right) \right) &=&  \NHBr{F, \NHBr{G,H}} +   \NHBr{G, \NHBr{H,F}}   \\
&   & - H \left[   \NHBr{F, \Reeb_{L, \Delta} \left(G \right)} - \NHBr{G, \Reeb_{L, \Delta} \left(F \right)} \right]
\end{eqnarray*}

\end{scriptsize}
Thus, we have that
\begin{scriptsize}

\begin{equation}\label{eq:otraotramas}
\begin{aligned}
&  X_{F}^{\Delta}\left( X_{G}^{\Delta} \left( H \right) \right) - X_{G}^{\Delta}\left( X_{H}^{\Delta} \left( H \right) \right) - X_{\NHBr{F,G}}^{\Delta}\left(H \right)  =\\
& \NHBr{F, \NHBr{G,H}} +   \NHBr{G, \NHBr{H,F}}  + \NHBr{H, \NHBr{F,G}}\\
& + H \Reeb_{L, \Delta}\left( \NHBr{F,G} \right)  - H \left[   \NHBr{F, \Reeb_{L, \Delta} \left(G \right)} - \NHBr{G, \Reeb_{L, \Delta} \left(F \right)} \right]
\end{aligned}
\end{equation}

\end{scriptsize}
The result now follows directly from \cref{eq:otramas} and \cref{eq:otraotramas}.
\end{proof}
\end{proposition}
We will now use this result to characterize an integrability condition on the constraint manifold. Assume that $\Delta$ is an integrable distribution on $Q$, i.e., the constraint Lagrangian system is semiholonomic. Let $\left( q^{i} \right)$ be a foliated system of coordinates on $Q$ associated to $\Delta$. Then
\begin{equation*}
\dfrac{\partial}{\partial q^{i} } \in \Delta, \ \ \ i \leq k.
\end{equation*}
So, we may assume that the constraint functions are $\Phi^a = \dd \dot{q}^{a}$ for $a  > k$. Then,
\begin{equation*}
Z_{a} = - W^{ia}\dfrac{\partial}{\partial \dot{q}^{i} }
\end{equation*}
Hence
\begin{equation*}
W_{ai}Z_{i} = - \dfrac{\partial}{\partial \dot{q}^{a} }\in \mathcal{S}, \ \ a > k.
\end{equation*}
On the other hand, it is obvious that
$$\dfrac{\partial}{\partial z },\dfrac{\partial}{\partial q^{j} }, \dfrac{\partial}{\partial \dot{q}^{i} }\in T \left( \Delta \times \RR \right),$$
for any $j$ and $i \leq k$. Then, 
$$ \mathcal{P} \left( \dfrac{\partial}{\partial \dot{q}^{i} } \right) = \dfrac{\partial}{\partial \dot{q}^{i} }, \ \ \mathcal{P} \left( \dfrac{\partial}{\partial q^{j} } \right) = \dfrac{\partial}{\partial q^{j} }, \ \ \mathcal{P} \left( \dfrac{\partial}{\partial z } \right) = \dfrac{\partial}{\partial z }, \ \ \mathcal{P} \left( \dfrac{\partial}{\partial \dot{q}^{a} } \right) = 0,$$
for all $j$, $i \leq k$ and $a >k$. We only have to use these equatlities to check that the Jacobi is satisfied in these coordinates, i.e., the integrability of the constraint manifold implies that the nonholonomic bracket induces a Jacobi structure on $\Delta \times \RR$. We may in fact prove that these two statements are equivalent.
\begin{theorem}
The constraint Lagrangian system $\left( L , \Delta \right)$ is semiholonomic if, and only if, the nonholonomic bracket satisfies the Jacobi identity.
\begin{proof}
Let $\Phi^{a}$ be the constraint functions. Consider $\tilde{\Phi}^{a}=\Phi^{a}_{i}{\dd q}^{i}$ the associated $1-$forms generating $\ann{\Delta^{l}}$. Then,
$$\mathcal{P}^{*}\tilde{\Phi}^{a} = \tilde{\Phi}^{a}, \ \forall a.$$
This is a direct consequence of that $\mathcal{P}^{*}\dd q^{i} = \dd q^{i}$ for all $i$.
Let us fix $H \in \Cont^\infty(TQ \times \RR)$. Taking into account that $\mathcal{P}^{*} \dd H \in \ann{\mathcal{S}}$, we have that
\begin{eqnarray*}
0 & = & \mathcal{P}^{*} \dd H  \left(  Z_{a} \right)\\
&=&  \mathcal{P}^{*} \dd H  \left(  \sharp_{L} \left( \tilde{\Phi}^{a}\right) \right)\\
&=& \mathcal{P}^{*} \dd H  \left(  \sharp_{L} \left( \mathcal{P}^{*}\tilde{\Phi}^{a}\right) \right)
\end{eqnarray*}
Thus, by using \cref{eq:nonholhamiltonian_vf_contactsecond} and $\Reeb_{L, \Delta} \in \Delta^{l}$, we have that
$$ \tilde{\Phi}^{a} \left(  X_{H}^{\Delta} \right)= 0,$$
i.e., $X_{H}^{\Delta} \in \Delta^{l}$ for all $H \in \Cont^\infty(TQ \times \RR)$. Let be a (local) basis $\{ X_{b} = X_{b}^{i}\dfrac{\partial}{\partial q^{i}}\}$ of $\Delta$. Then, consider $\Lambda_{b}$ the local functions on $TQ \times \RR$ induced by the $1-$forms $X_{b}^{i}dq^{i}$. Hence, by taking into account that the correspondence $ H \mapsto X_{H}^{\Delta}$ is an isomorphism of vector spaces, we have that the family $\{ X_{\Lambda_{b}}^{\Delta}, X_{z}^{\Delta}\}$ is a (local) basis of $\Delta^{l}$ where $z$ is the natural projection of $TQ \times \RR$ onto $\RR$.\\
So, taking into account \cref{yotramasss}, we have that the distribution $\Delta^{l}$ is involutive.\\
Consider now an arbitrary vector field $X$ on $Q$. Then, there exists a (local) vector field $X^{l}$ on $TQ \times \RR$ which is $\left(\tau_{Q} \circ pr_{TQ \times \RR}\right)-$related with $X$, i.e., the diagram

\begin{large}
\begin{center}
 \begin{tikzcd}[column sep=huge,row sep=huge]
TQ \times \RR \arrow[d,"\tau_{Q} \circ pr_{TQ \times \RR}"] \arrow[r, "X^{l}"] &T \left( TQ \times \RR\right) \arrow[d, "T\left(\tau_{Q} \circ pr_{TQ \times \RR} \right)"] \\
 Q  \arrow[r,"X"] &TQ
\end{tikzcd}
\end{center}
\end{large}
is commutative. In fact, let us consider a (local) basis $\{\sigma^{i}\}$ of section of $ \tau_{Q} \circ pr_{TQ \times \RR}$. Then, we may construct $X^{l}$ as follows
$$ X^{l} \left(\lambda_{i} \sigma^{i}\left(q \right)\right) = \lambda_{i} T_{q}\sigma_{i} \left(X \left( q \right) \right),$$
for all $q$ in the domain of the basis. It is finally trivial to check that $X \in \Delta$ if, and only if, any $\left( \tau_{Q} \circ pr_{TQ \times \RR} \right)-$related vector field on $X^{l}$ on $TQ \times \RR$ $X^{l}$ with $X$ is in $\Delta^{l}$. Thus, $\Delta$ is also involutive and, therefore, integrable.

\end{proof}
\end{theorem}
Therefore, we have proved that the nonholonomic condition of the constraint Lagrangian system $\left(L, \Delta \right)$ may be checked by the Jacobi identity of the nonholonomic brackets.

\section{Example: Chaplygin's sledge}
We will present here a model for a sledge with homogeneous and isotropic Rayleigh dissipation.

A detailed study of the Chaplygin's sledge may be found in \cite{Chaplygin1949,MOSHCHUK1987}. The nonholonomic character of this example has been investigated in \cite{deLeon1996d}.
The system, which is described in \cite{MOSHCHUK1987}, 
has a configuration space $Q= \RR^2 \times S^1$, which coordinates $(x,y)$, describing the position of a blade and and angle $\theta$, which describes its rotation. We  addeed an extra term $\gamma z$, which accounts for friction with the the medium following the model of Rayleigh dissipation~\cite{Goldstein2006,Bravetti2019,Bravetti2017}. So, the resulting Lagrangian is given by
\begin{equation}
\begin{aligned}
    L=&
    \frac{1}{2} \, {\left({\left({\alpha} \cos\left({\phi}\right) - {\beta} \sin\left({\phi}\right)\right)} {\dot{\phi}} + {\dot{y}}\right)}^{2} + \\ & \frac{1}{2} \, {\left({\left({\beta} \cos\left({\phi}\right) + {\alpha} \sin\left({\phi}\right)\right)} {\dot{\phi}} - {\dot{x}}\right)}^{2} + {\dot{\phi}}^{2} + {\gamma} z
\end{aligned}
\end{equation}
where $\gamma$ is the friction coefficient of the carriage with the medium, and $(\alpha,\beta)$ is the position of the sledge center of gravity $C$ in the reference frame formed by the axes $A$ and $B$ (see~\cref{fig:sledge}). The units are normalized so that the mass and the radius of inertia of the sledge is $1$. 

\begin{figure}
    \centering
    \def\svgwidth{0.8\textwidth}
    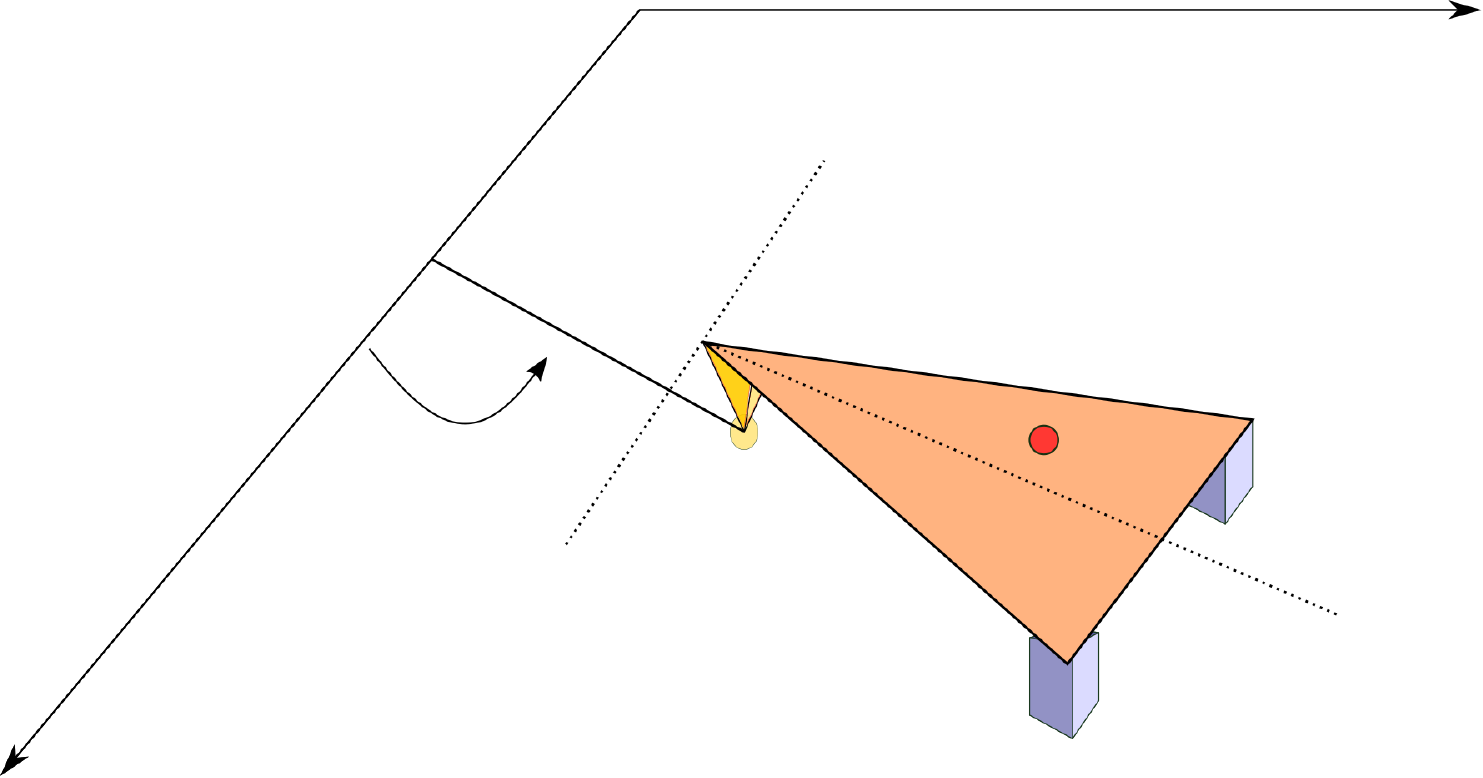
    \caption{}\label{fig:sledge}
\end{figure}

Thus, the contact form $\eta_{L}$ is written as follows
\begin{center}
\begin{align*}
\eta_L &= \left( {\left({\beta} \cos\left({\phi}\right) + {\alpha} \sin\left({\phi}\right)\right)} {\dot{\phi}} - {\dot{x}} \right) \mathrm{d} x + \left( -{\left({\alpha} \cos\left({\phi}\right) - {\beta} \sin\left({\phi}\right)\right)} {\dot{\phi}} - {\dot{y}} \right) \mathrm{d} y + \\ & 
\left( -{\left({\alpha}^{2} + {\beta}^{2} + 2\right)} {\dot{\phi}} + {\left({\beta} \cos\left({\phi}\right) + {\alpha} \sin\left({\phi}\right)\right)} {\dot{x}} - {\left({\alpha} \cos\left({\phi}\right) - {\beta} \sin\left({\phi}\right)\right)} {\dot{y}} \right) \mathrm{d} {\phi} \\ & +\mathrm{d} z\end{align*}
\end{center}
Furthermore, it is easy to check that $\Reeb_{L} = \dfrac{\partial}{\partial z}$.\\
Observe that the Hessian matrix $\left(W_{ij} \right) = \left(\frac{\partial^2 L}{\partial \dot{q}^i \partial \dot{q}^{j}}\right)$ is just
\[
\left(\begin{array}{rrr}
1 & 0 & -{\beta} \cos\left({\phi}\right) - {\alpha} \sin\left({\phi}\right) \\
0 & 1 & {\alpha} \cos\left({\phi}\right) - {\beta} \sin\left({\phi}\right) \\
-{\beta} \cos\left({\phi}\right) - {\alpha} \sin\left({\phi}\right) & {\alpha} \cos\left({\phi}\right) - {\beta} \sin\left({\phi}\right) & {\alpha}^{2} + {\beta}^{2} + 2
\end{array}\right),\]
which has determinant $2$, hence the system is regular.

Next, we will calculate the distribution $\mathcal{S}$, which, in this case, has rank~$1$. We find the generating vector field $Z_{1}$ (see \cref{14}) by using formula (\ref{15}):

    \begin{align*}
     Z_{1} &= 
     \left( -\frac{1}{2} \, {\alpha} {\beta} \cos\left({\phi}\right) - \frac{1}{2} \, {\left({\alpha}^{2} + 2\right)} \sin\left({\phi}\right) \right) \frac{\partial}{\partial {\dot{x}} } \\ & \quad
     + \left( -\frac{1}{2} \, {\alpha} {\beta} \sin\left({\phi}\right) + \frac{1}{2} \, {\left({\alpha}^{2} + 2\right)} \cos\left({\phi}\right) \right) \frac{\partial}{\partial {\dot{y}} } -\frac{1}{2} \, {\alpha} \frac{\partial}{\partial {\dot{\phi}} }
    \end{align*}
Thus, $\mathcal{S} = \gen{ Z_{1}}$. We will now prove the \textit{existence} and \textit{uniqueness} of the solutions of the problem by using \cref{lamatrizcdecarajo}. Thus, we should study if the matrix $\left(\mathcal{C}_{ab} \right) = \left(- W^{ik}\psi^{b}_{k}\psi^{a}_{i} \right)$ is regular. A calculation shows that the matrix $\left(\mathcal{C}_{ab}\right)$, which in this case is a real number:
\[
C_{ab}=
 -1/2\alpha^2 - 1
\]
Therefore, we may conclude that our system has the property of \textit{uniqueness} and \textit{existence of solutions}.

The projector $\mathcal{P}$ (\cref{projectors2}) has the following nonzero components
\begin{equation*}
    \begin{array}{lcl} \mathcal{P}_{ \phantom{\, x} \, x }^{ \, x \phantom{\, x} } & = & 1 \\ \mathcal{P}_{ \phantom{\, y} \, y }^{ \, y \phantom{\, y} } & = & 1 \\ \mathcal{P}_{ \phantom{\, {\phi}} \, {\phi} }^{ \, {\phi} \phantom{\, {\phi}} } & = & 1 \\ \mathcal{P}_{ \phantom{\, {\dot{x}}} \, {\phi} }^{ \, {\dot{x}} \phantom{\, {\phi}} } & = & -\frac{1}{4} \, {\left({\alpha}^{3} + 2 \, {\alpha}\right)} {\beta} {\dot{x}} - \frac{1}{4} \, {\left({\alpha}^{4} + 4 \, {\alpha}^{2} + 4\right)} {\dot{y}} \\ \mathcal{P}_{ \phantom{\, {\dot{x}}} \, {\dot{x}} }^{ \, {\dot{x}} \phantom{\, {\dot{x}}} } & = & -\frac{1}{4} \, {\left({\alpha}^{3} + 2 \, {\alpha}\right)} {\beta} \cos\left({\phi}\right) \sin\left({\phi}\right) - \frac{1}{4} \, {\left({\alpha}^{4} + 4 \, {\alpha}^{2} + 4\right)} \sin\left({\phi}\right)^{2} + 1 \\ \mathcal{P}_{ \phantom{\, {\dot{x}}} \, {\dot{y}} }^{ \, {\dot{x}} \phantom{\, {\dot{y}}} } & = & \frac{1}{4} \, {\left({\alpha}^{3} + 2 \, {\alpha}\right)} {\beta} \cos\left({\phi}\right)^{2} + \frac{1}{4} \, {\left({\alpha}^{4} + 4 \, {\alpha}^{2} + 4\right)} \cos\left({\phi}\right) \sin\left({\phi}\right) \\ \mathcal{P}_{ \phantom{\, {\dot{y}}} \, {\phi} }^{ \, {\dot{y}} \phantom{\, {\phi}} } & = & -\frac{1}{4} \, {\left({\alpha}^{3} + 2 \, {\alpha}\right)} {\beta} {\dot{y}} + \frac{1}{4} \, {\left({\alpha}^{4} + 4 \, {\alpha}^{2} + 4\right)} {\dot{x}} \\ \mathcal{P}_{ \phantom{\, {\dot{y}}} \, {\dot{x}} }^{ \, {\dot{y}} \phantom{\, {\dot{x}}} } & = & -\frac{1}{4} \, {\left({\alpha}^{3} + 2 \, {\alpha}\right)} {\beta} \sin\left({\phi}\right)^{2} + \frac{1}{4} \, {\left({\alpha}^{4} + 4 \, {\alpha}^{2} + 4\right)} \cos\left({\phi}\right) \sin\left({\phi}\right) \\ \mathcal{P}_{ \phantom{\, {\dot{y}}} \, {\dot{y}} }^{ \, {\dot{y}} \phantom{\, {\dot{y}}} } & = & \frac{1}{4} \, {\left({\alpha}^{3} + 2 \, {\alpha}\right)} {\beta} \cos\left({\phi}\right) \sin\left({\phi}\right) - \frac{1}{4} \, {\left({\alpha}^{4} + 4 \, {\alpha}^{2} + 4\right)} \cos\left({\phi}\right)^{2} + 1 \\ \mathcal{P}_{ \phantom{\, {\dot{\phi}}} \, {\phi} }^{ \, {\dot{\phi}} \phantom{\, {\phi}} } & = & -\frac{1}{4} \, {\left({\alpha}^{3} + 2 \, {\alpha}\right)} {\dot{x}} \cos\left({\phi}\right) - \frac{1}{4} \, {\left({\alpha}^{3} + 2 \, {\alpha}\right)} {\dot{y}} \sin\left({\phi}\right) \\ \mathcal{P}_{ \phantom{\, {\dot{\phi}}} \, {\dot{x}} }^{ \, {\dot{\phi}} \phantom{\, {\dot{x}}} } & = & -\frac{1}{4} \, {\left({\alpha}^{3} + 2 \, {\alpha}\right)} \sin\left({\phi}\right) \\ \mathcal{P}_{ \phantom{\, {\dot{\phi}}} \, {\dot{y}} }^{ \, {\dot{\phi}} \phantom{\, {\dot{y}}} } & = & \frac{1}{4} \, {\left({\alpha}^{3} + 2 \, {\alpha}\right)} \cos\left({\phi}\right) \\ \mathcal{P}_{ \phantom{\, {\dot{\phi}}} \, {\dot{\phi}} }^{ \, {\dot{\phi}} \phantom{\, {\dot{\phi}}} } & = & 1 \\ \mathcal{P}_{ \phantom{\, z} \, z }^{ \, z \phantom{\, z} } & = & 1 \end{array}
\end{equation*}

The dynamics of the unconstrained system is then given by
\begin{equation}
\begin{aligned}
    \Gamma_L &= {\dot{x}} \frac{\partial}{\partial x } + {\dot{y}} \frac{\partial}{\partial y } + {\dot{\phi}} \frac{\partial}{\partial {\phi} } + \left( {\left({\alpha} \cos\left({\phi}\right) - {\beta} \sin\left({\phi}\right)\right)} {\dot{\phi}}^{2} + {\gamma} {\dot{x}} \right) \frac{\partial}{\partial {\dot{x}} } \\ & \quad
    + \left( {\left({\beta} \cos\left({\phi}\right) + {\alpha} \sin\left({\phi}\right)\right)} {\dot{\phi}}^{2} + {\gamma} {\dot{y}} \right) \frac{\partial}{\partial {\dot{y}} } + {\gamma} {\dot{\phi}} \frac{\partial}{\partial {\dot{\phi}} } \\ & \quad
    + \bigl( \frac{1}{2} \, {\left({\alpha}^{2} + {\beta}^{2} + 2\right)} {\dot{\phi}}^{2} - {\left({\beta} \cos\left({\phi}\right) + {\alpha} \sin\left({\phi}\right)\right)} {\dot{\phi}} {\dot{x}} \\ & \quad 
    +  {\left({\alpha} \cos\left({\phi}\right) - {\beta} \sin\left({\phi}\right)\right)} {\dot{\phi}} {\dot{y}} + \frac{1}{2} \, {\dot{x}}^{2} + \frac{1}{2} \, {\dot{y}}^{2} + {\gamma} z \bigl) \frac{\partial}{\partial z }
\end{aligned}
\end{equation}

The constrained dynamics $\Gamma_{L,\Delta}$ is given by projecting $\xi_{L}$:
\begin{equation}
\begin{aligned}
\Gamma_{L,\Delta} =& {\dot{x}} \frac{\partial}{\partial x } + {\dot{y}} \frac{\partial}{\partial y } + {\dot{\phi}} \frac{\partial}{\partial {\phi} } \\ &+ 
\biggl( \frac{1}{4} \, {\left({\left({\alpha}^{4} + 4 \, {\alpha}^{2}\right)} {\beta} \sin\left({\phi}\right) + {\left({\left({\alpha}^{3} + 2 \, {\alpha}\right)} {\beta}^{2} + 4 \, {\alpha}\right)} \cos\left({\phi}\right)\right)} {\dot{\phi}}^{2} \\& \quad - \frac{1}{4} \, {\left({\alpha}^{4} + 4 \, {\alpha}^{2} + 4\right)} {\dot{\phi}} {\dot{y}} - \frac{1}{4} \, {\left({\left({\alpha}^{3} + 2 \, {\alpha}\right)} {\beta} {\dot{\phi}} - 4 \, {\gamma}\right)} {\dot{x}} \biggl)   \frac{\partial}{\partial {\dot{x}} } \\& 
+ \biggl( -\frac{1}{4} \, {\left({\left({\alpha}^{4} + 4 \, {\alpha}^{2}\right)} {\beta} \cos\left({\phi}\right) - {\left({\left({\alpha}^{3} + 2 \, {\alpha}\right)} {\beta}^{2} + 4 \, {\alpha}\right)} \sin\left({\phi}\right)\right)} {\dot{\phi}}^{2} \\ & \quad + \frac{1}{4} \, {\left({\alpha}^{4} + 4 \, {\alpha}^{2} + 4\right)} {\dot{\phi}} {\dot{x}} - \frac{1}{4} \, {\left({\left({\alpha}^{3} + 2 \, {\alpha}\right)} {\beta} {\dot{\phi}} - 4 \, {\gamma}\right)} {\dot{y}} \biggl)  \frac{\partial}{\partial {\dot{y}} } \\ &
+ \left( \frac{{\left({\alpha}^{3} + 2 \, {\alpha}\right)} {\beta} {\dot{\phi}}^{2} \cos\left({\phi}\right) - {\left({\alpha}^{3} + 2 \, {\alpha}\right)} {\dot{\phi}} {\dot{x}} + 4 \, {\gamma} {\dot{\phi}} \cos\left({\phi}\right)}{4 \, \cos\left({\phi}\right)} \right) \frac{\partial}{\partial {\dot{\phi}} } \\ &
+ \biggl( \frac{1}{2} \, {\left({\alpha}^{2} + {\beta}^{2} + 2\right)} {\dot{\phi}}^{2} - {\left({\beta} \cos\left({\phi}\right) + {\alpha} \sin\left({\phi}\right)\right)} {\dot{\phi}} {\dot{x}} \\ & \quad + {\left({\alpha} \cos\left({\phi}\right) - {\beta} \sin\left({\phi}\right)\right)} {\dot{\phi}} {\dot{y}} + \frac{1}{2} \, {\dot{x}}^{2} + \frac{1}{2} \, {\dot{y}}^{2} + {\gamma} z \biggl) \frac{\partial}{\partial z }
\end{aligned}
\end{equation}

The nonzero nonholonomic brackets are given between the coordinate functions and the constant $1$ (which are sufficient to characterise the almost-Jacobi algebra) are given by:
\begin{equation*}
    \begin{array}{lcl}
    \NHBr{1,z} &=& -1 \\
    \NHBr{ q^i , z } & = & q^i \\
    \NHBr{ \dot{q}^i , z } & = & 2 \dot{q}^i \\
    \NHBr{x, {\dot{x}}} & = & \frac{1}{8} \, {\alpha}^{6} + \frac{3}{4} \, {\alpha}^{4} + \frac{1}{4} \, {\left({\alpha}^{5} + 4 \, {\alpha}^{3}\right)} {\beta} \cos\left({\phi}\right) \sin\left({\phi}\right) \\ && - \frac{1}{8} \, {\left({\alpha}^{6} + 6 \, {\alpha}^{4} - {\left({\alpha}^{4} + 2 \, {\alpha}^{2} - 4\right)} {\beta}^{2} + 8 \, {\alpha}^{2} + 8\right)} \cos\left({\phi}\right)^{2} + {\alpha}^{2} \\ 
    \NHBr{x,\dot{y}} & = & -\frac{1}{4} \, {\left({\alpha}^{5} + 4 \, {\alpha}^{3}\right)} {\beta} \cos\left({\phi}\right)^{2}  + \frac{1}{8} \, {\left({\alpha}^{5} + 4 \, {\alpha}^{3}\right)} {\beta}  \\ && - \frac{1}{8} \, {\left({\alpha}^{6} + 6 \, {\alpha}^{4} - {\left({\alpha}^{4} + 2 \, {\alpha}^{2} - 4\right)} {\beta}^{2} + 8 \, {\alpha}^{2} + 8\right)} \cos\left({\phi}\right) \sin\left({\phi}\right)\\
    \NHBr{x,\dot{\phi}} & = & \frac{1}{8} \, {\left({\alpha}^{4} + 2 \, {\alpha}^{2} - 4\right)} {\beta} \cos\left({\phi}\right) + \frac{1}{8} \, {\left({\alpha}^{5} + 4 \, {\alpha}^{3}\right)} \sin\left({\phi}\right) \\ 
     \NHBr{y, {\dot{x}} } & = & -\frac{1}{4} \, {\left({\alpha}^{5} + 4 \, {\alpha}^{3}\right)} {\beta} \cos\left({\phi}\right)^{2} + \frac{1}{8} \, {\left({\alpha}^{5} + 4 \, {\alpha}^{3}\right)} {\beta} \\ && - \frac{1}{8} \, {\left({\alpha}^{6} + 6 \, {\alpha}^{4} - {\left({\alpha}^{4} + 2 \, {\alpha}^{2} - 4\right)} {\beta}^{2} + 8 \, {\alpha}^{2} + 8\right)} \cos\left({\phi}\right) \sin\left({\phi}\right) \\
    \NHBr{ y, {\dot{y}} } & = & \frac{1}{8} \, {\alpha}^{6} + \frac{3}{4} \, {\alpha}^{4} - \frac{1}{4} \, {\left({\alpha}^{5} + 4 \, {\alpha}^{3}\right)} {\beta} \cos\left({\phi}\right) \sin\left({\phi}\right) \\ && - \frac{1}{8} \, {\left({\alpha}^{6} + 6 \, {\alpha}^{4} - {\left({\alpha}^{4} + 2 \, {\alpha}^{2} - 4\right)} {\beta}^{2} + 8 \, {\alpha}^{2} + 8\right)} \sin\left({\phi}\right)^{2} + {\alpha}^{2} \\
    \NHBr{y, {\dot{\phi}} } & = & \frac{1}{8} \, {\left({\alpha}^{4} + 2 \, {\alpha}^{2} - 4\right)} {\beta} \sin\left({\phi}\right) - \frac{1}{8} \, {\left({\alpha}^{5} + 4 \, {\alpha}^{3}\right)} \cos\left({\phi}\right) \\ 
    \NHBr{ {\phi} , {\dot{x}} } & = & \frac{1}{8} \, {\left({\alpha}^{4} + 2 \, {\alpha}^{2} - 4\right)} {\beta} \cos\left({\phi}\right) + \frac{1}{8} \, {\left({\alpha}^{5} + 4 \, {\alpha}^{3}\right)} \sin\left({\phi}\right) \\ 
    \NHBr{{\phi} , {\dot{y}} } & = & \frac{1}{8} \, {\left({\alpha}^{4} + 2 \, {\alpha}^{2} - 4\right)} {\beta} \sin\left({\phi}\right) - \frac{1}{8} \, {\left({\alpha}^{5} + 4 \, {\alpha}^{3}\right)} \cos\left({\phi}\right) \\ 
    \NHBr{ {\phi}, {\dot{\phi}} } & = & \frac{1}{8} \, {\alpha}^{4} + \frac{1}{4} \, {\alpha}^{2} - \frac{1}{2}\\ 
    \NHBr{  {\dot{x}} , {\dot{y}} } & = & \frac{2 \, {\left({\alpha}^{2} + 2\right)} {\beta} {\dot{x}}}{8 \, \cos\left({\phi}\right)}  - {\left({\alpha}^{6} + 4 \, {\alpha}^{4} + {\left({\alpha}^{4} + 4 \, {\alpha}^{2}\right)} {\beta}^{2}\right)} {\dot{\phi}}  \\
    \NHBr{ {\dot{x}}, {\dot{\phi}} } & = & -\frac{1}{8} \, {\left({\left({\alpha}^{4} + 4 \, {\alpha}^{2}\right)} {\beta} \sin\left({\phi}\right) - {\left({\alpha}^{5} + 2 \, {\alpha}^{3} - 4 \, {\alpha}\right)} \cos\left({\phi}\right)\right)} {\dot{\phi}} \\ && + \frac{1}{4} \, {\left({\alpha}^{2} + 2\right)} {\dot{y}} \\
    \NHBr{{\dot{y}}, {\dot{\phi}} } & = & \frac{1}{8} \, {\left({\left({\alpha}^{4} + 4 \, {\alpha}^{2}\right)} {\beta} \cos\left({\phi}\right) + {\left({\alpha}^{5} + 2 \, {\alpha}^{3} - 4 \, {\alpha}\right)} \sin\left({\phi}\right)\right)} {\dot{\phi}} \\ && - \frac{1}{4} \, {\left({\alpha}^{2} + 2\right)} {\dot{x}} 
    \end{array}
\end{equation*}

\section*{Acknowledgments}

This work has been partially supported by MINECO Grants MTM2016-76-072-P and the ICMAT Severo Ochoa projects SEV-2011-0087 and SEV-2015-0554. M. Lainz and V.M.~Jim{\'e}nez wishes to thank MINECO for a FPI-PhD Position.

\printbibliography

\end{document}

%% file: Trineo.eps_tex
\begingroup%
  \makeatletter%
  \providecommand\color[2][]{%
    \errmessage{(Inkscape) Color is used for the text in Inkscape, but the package 'color.sty' is not loaded}%
    \renewcommand\color[2][]{}%
  }%
  \providecommand\transparent[1]{%
    \errmessage{(Inkscape) Transparency is used (non-zero) for the text in Inkscape, but the package 'transparent.sty' is not loaded}%
    \renewcommand\transparent[1]{}%
  }%
  \providecommand\rotatebox[2]{#2}%
  \newcommand*\fsize{\dimexpr\f@size pt\relax}%
  \newcommand*\lineheight[1]{\fontsize{\fsize}{#1\fsize}\selectfont}%
  \ifx\svgwidth\undefined%
    \setlength{\unitlength}{426.22554268bp}%
    \ifx\svgscale\undefined%
      \relax%
    \else%
      \setlength{\unitlength}{\unitlength * \real{\svgscale}}%
    \fi%
  \else%
    \setlength{\unitlength}{\svgwidth}%
  \fi%
  \global\let\svgwidth\undefined%
  \global\let\svgscale\undefined%
  \makeatother%
  \begin{picture}(1,0.52338418)%
    \lineheight{1}%
    \setlength\tabcolsep{0pt}%
    \put(0,0){\includegraphics[width=\unitlength]{Trineo.eps}}%
    \put(-3.08718005,-1.75727711){\color[rgb]{0,0,0}\makebox(0,0)[lt]{\begin{minipage}{0.08061713\unitlength}\raggedright \end{minipage}}}%
    \put(0.22378077,0.31971706){\color[rgb]{0,0,0}\makebox(0,0)[lt]{\lineheight{1.25}\smash{\begin{tabular}[t]{l}\textit{x}\end{tabular}}}}%
    \put(0.92410236,0.48114413){\color[rgb]{0,0,0}\makebox(0,0)[lt]{\lineheight{1.25}\smash{\begin{tabular}[t]{l}\textit{y}\end{tabular}}}}%
    \put(0.52663623,-1.19663378){\color[rgb]{0,0,0}\makebox(0,0)[lt]{\begin{minipage}{0.10700826\unitlength}\raggedright \end{minipage}}}%
    \put(0.27473382,0.19369511){\color[rgb]{0,0,0}\makebox(0,0)[lt]{\lineheight{1.25}\smash{\begin{tabular}[t]{l}\textit{$\phi$}\end{tabular}}}}%
    \put(0.7234078,0.22438715){\color[rgb]{0,0,0}\makebox(0,0)[lt]{\lineheight{1.25}\smash{\begin{tabular}[t]{l}\textit{C}\end{tabular}}}}%
    \put(0.48081159,0.35392934){\color[rgb]{0,0,0}\makebox(0,0)[lt]{\lineheight{1.25}\smash{\begin{tabular}[t]{l}\textit{B}\end{tabular}}}}%
    \put(0.82939685,0.09720015){\color[rgb]{0,0,0}\makebox(0,0)[lt]{\lineheight{1.25}\smash{\begin{tabular}[t]{l}\textit{A}\end{tabular}}}}%
    \put(0.47125671,0.18536793){\color[rgb]{0,0,0}\makebox(0,0)[lt]{\lineheight{1.25}\smash{\begin{tabular}[t]{l}\textit{(x,y)}\end{tabular}}}}%
  \end{picture}%
\endgroup%